\documentclass[runningheads,a4paper]{llncs}
\usepackage[letterpaper,margin=1in]{geometry}
\usepackage{graphicx}
\usepackage{amsmath}
\usepackage{amssymb}
\usepackage{tikz}
\usepackage{float}
\usepackage[ruled,vlined,english]{algorithm2e}
\usepackage{subcaption}
\usetikzlibrary{arrows}

\let\doendproof\endproof
\renewcommand\endproof{~\hfill$\blacksquare$\doendproof}

\title{The maximum time of 2-neighbour bootstrap percolation in grid graphs and some parameterized results}

\author{Thiago Marcilon \and Rudini Sampaio}

\institute{Dept. Computa\c c\~ao, Universidade Federal do Cear\'a, Fortaleza, Brazil\\\email{\{thiagomarcilon,rudini\}@lia.ufc.br}}

\begin{document}
\mainmatter
\maketitle

\begin{abstract}
In 2-neighborhood bootstrap percolation on a graph $G$, an infection spreads according to the following deterministic rule: infected vertices of $G$ remain infected forever and in consecutive rounds healthy vertices with at least two already infected neighbors become infected. Percolation occurs if eventually every vertex is infected. The maximum time $t(G)$ is the maximum number of rounds needed to eventually infect the entire vertex set. In 2013, it was proved by Benevides et al \cite{eurocomb13} that $t(G)$ is NP-hard for planar graphs and that deciding whether $t(G)\geq k$ is polynomial time solvable for $k\leq 2$, but is NP-complete for $k\geq 4$. They left two open problems about the complexity for $k=3$ and for planar bipartite graphs. In 2014,  we solved the first problem\cite{wg2014}. In this paper, we solve the second one by proving that $t(G)$ is NP-complete even in grid graphs with maximum degree 3. We also prove that $t(G)$ is polynomial time solvable for solid grid graphs with maximum degree 3. Moreover, we prove that the percolation time problem is W[1]-hard on the treewidth of the graph, but it is fixed parameter tractable with parameters treewidth$+k$ and maxdegree$+k$.
\end{abstract}


% ----------------------------------------------------------------------------------
% ----------------------------------------------------------------------------------
% ----------------------------------------------------------------------------------
% ----------------------------------------------------------------------------------
\section{Introduction}

We consider a problem in which an infection spreads over the vertices of a connected simple graph $G$ following a deterministic spreading rule in such a way that an infected vertex will remain infected forever. Given a set $S \subseteq V(G)$ of initially infected vertices, we build a sequence $S_0, S_1, S_2, \ldots$ in which $S_0=S$ and $S_{i+1}$ is obtained from $S_i$ using such spreading rule.

Under $r$-neighbor bootstrap percolation on a graph $G$, the spreading rule is a threshold rule in which $S_{i+1}$ is obtained from $S_i$ by adding to it the vertices of $G$ which have at least $r$ neighbors in $S_i$. We say that a set $S$ infects a vertex $v$ at time $i$ if  $v \in S_i \setminus S_{i-1}$.
Let, for any set of vertices $S$ and vertex $v$ of $G$, $t_r(G,S,v)$ be the minimum $t$ such that $v$ belongs to $S_t$ or, if there is no $t$ such that $v$ belongs to $S_t$, then $t_r(G,S,v) = \infty$. 
Also, we say that a set $S_0$ infects $G$, or that $S_0$ is a percolating set of $G$, if eventually every vertex of $G$ becomes infected, that is, there exists a $t$ such that $S_t = V(G)$. If $S$ is a percolating set of $G$, then we define $t_r(G,S)$ as the minimum $t$ such that $S_t = V(G)$. Also, define the {\em percolation time of $G$} as $t_r(G) = \max \{t_r(G,S) : S \text{ is a percolating set of } G\}$. In this paper, we shall focus on the case where $r=2$ and in such case we omit the subscript of the notations $t_r(G,S)$ and $t_r(G)$. Also, from the notation $t(G,S)$ and $t(G,S,v)$, when the parameter $G$ is clear from context, it will be omitted.

Bootstrap percolation was introduced by Chalupa, Leath and Reich \cite{chalupa} as a model
for certain interacting particle systems in physics. Since then it has found applications
in clustering phenomena, sandpiles \cite{sandpiles}, and many other areas of statistical physics, as well as in neural networks \cite{neural2} and computer science \cite{cs1}.

There are two broad classes of questions one can ask about bootstrap percolation. The
first, and the most extensively studied, is what happens when the initial configuration $S_0$ is chosen randomly under some probability distribution? For example, vertices are included in $S_0$ independently with some fixed probability $p$. One would like to know how likely percolation is to occur, and if it does occur, how long it takes. The answer to these questions is now well understood for various types of graphs \cite{holroyd,balogh,balogh2,balogh4,bollobasHolmgrenSmithUzzell}.

The second broad class of questions is the one of extremal questions. For example, what is the smallest or largest size of a percolating set with a given property? The size of the smallest percolating set in the $d$-dimensional grid, $[n]^d$, was studied by Pete and a summary can be found in \cite{BaloghPete}. Morris \cite{morris} and Riedl \cite{riedl} studied the maximum size of minimal percolating sets on the square grid $[n]^2$ and the hypercube $\{0,1\}^d$, respectively, answering a question posed by Bollob\'as. However, the problem of finding the smallest percolating set is NP-hard even on subgraphs of the square grid \cite{jayme13} and it is APX-hard even for bipartite graphs with maximum degree four \cite{waoa-13}. Moreover, it is hard \cite{ningchen} to approximate within a ratio $O(2^{\log^{1-\varepsilon}n})$, for any $\varepsilon>0$, unless $NP\subseteq DTIME(n^{polylog(n)})$.

Another type of question is: what is the minimum or maximum time that percolation can take, given that $S_0$ satisfies certain properties? Recently, Przykucki \cite{Przykucki} determined the precise value of the maximum percolation time on the hypercube $2^{[n]}$ as a function of $n$, and Benevides and Przykucki \cite{fabricio1,fabriciob1} have similar results for the square grid $[n]^2$, also answering a question posed by Bollob\'as. In particular, they have a polynomial time dynamic programming algorithm to compute the maximum percolation time on rectangular grids \cite{fabricio1}.

Here, we consider the decision version of the Percolation Time Problem, as stated below.

\vspace{5 pt}%\bigskip
\noindent{\sc Percolation Time} \\
{\em Input:} A graph $G$ and an integer $k$. \\
{\em Question:} Is $t(G) \geq k$?
\vspace{5 pt}%\bigskip

In 2013, Benevides et.al. \cite{eurocomb13}, among other results, proved that the Percolation Time Problem is polynomial time solvable for $k\leq 2$, but is NP-complete for $k\geq 4$ and, when restricted to bipartite graphs, it is NP-complete for $k\geq 7$. Moreover, it was proved that the Percolation Time Problem is NP-complete for planar graphs. They left three open questions about the complexity for $k=3$ in general graphs, the complexity for $3\leq k\leq 6$ in bipartite graphs and the complexity for planar bipartite graphs.

In 2014, the first and the second questions were solved \cite{wg2014}: it was proved that the Percolation Time Problem is $O(mn^5)$-time solvable for $k=3$ in general graphs and, when restricted to bipartite graphs, it is $O(mn^3)$-time solvable for $k=3$, it is $O(m^2n^9)$-time solvable for $k=4$ and it is NP-complete for $k\geq 5$.

In this paper, we solve the third question of \cite{eurocomb13}.
We prove that the Percolation Time Problem is NP-complete for planar bipartite graphs. In fact, we prove a stronger result: the NP-completeness for grid graphs, which are induced subgraphs of grids, with maximum degree 3.

There are NP-hard problems in grid graphs which are polynomial time solvable for solid grid graphs. For example, the Hamiltonian cycle problem is NP-complete for grid graphs \cite{itai82}, but it is polynomial time solvable for solid grid graphs \cite{lenhart97}. Motivated by the work of \cite{fabricio1} for rectangular grids, we obtain in this paper a polynomial time algorithm for solid grid graphs with maximum degree 3. %However, the problem for general solid grid graphs is still open.

Finally, we prove several complexity results for $t(G)$ in graphs with bounded maximum degree and bounded treewidth, some of which implies fixed parameter tractable algorithms for the Percolation Time Problem. Moreover, we obtain polynomial time algorithms for $(q,q-4)$-graphs, for any fixed $q$, which are the graphs such that every subset of at most $q$ vertices induces at most $q-4$ $P_4$'s. Cographs and $P_4$-sparse graphs are exactly the $(4,0)$-graphs and the $(5,1)$-graphs, respectively. These algorithms are fixed parameter tractable on the parameter $q$.

\section{Percolation Time Problem in grid graphs with $\Delta = 3$}
\label{griddelta3}

In this section, we prove that the Percolation Time Problem is NP-complete in grid graphs with maximum degree $\Delta=3$.
We also show that, when the graph is a grid graph with $\Delta=3$ and $k=O(\log n)$, the Percolation Time Problem can be solved in polynomial time. But, first, let us define a $S$-infection path and, then, prove two lemmas that will be useful in the proofs.

Let $S$ be a percolating set. A path $P = v_0,v_1,\hdots,v_n$ is a $S$-\emph{infection path} if, for every $0\leq i\leq n-1$, $t(S,v_i) < t(S,v_{i+1})$.
Notice that, if $t(S,v)=k$, then there is a $S$-infection path $v_0,v_1,\hdots,v_k=v$, where $t(S,v_i)=i$ for each $0\leq i\leq k$.
%Due to space restrictions, its proof will be omitted.

Roughly speaking, the next lemma shows that, if the $S$-infection time of a vertex $v$ is decreased by the inclusion of a vertex $v'$ to $S$, then there is an infection path starting at $v'$ and ending at $v$.
\begin{lemma}
\label{lemacaminho} 
Let $S$ be a subset of $V(G)$, $v,v'$ be vertices of $G$ and $S' = S\cup\{v'\}$. If $t(S',v)<t(S,v)$, then there is a $S'$-infection path starting in $v'$ and ending in $v$.
\end{lemma}

\begin{proof}
Notice that, since $t(S',v)<t(S,v)$, then $t(S',v)<\infty$, i.e., $S'$ infects $v$. Also, $t(S',v)$ cannot be 0; otherwise,  $v\in S'$ and then $S=S'$ and $t(S,v)=t(S',v)$, a contradiction. Thus $t(S',v) \geq 1$.

Let us prove by induction on $t(S',v)$ that there is a $S'$-infection path starting in $v'$ and ending in $v$. For $t(S',v) = 1$, we have that $v$ must be neighbor of two vertices in $S'$, where one of these two vertices must be $v'$ because, otherwise, $t(S,v) = 1 = t(S',v)$. Thus, the path $v',v$ is a $S'$-infection path.

Now, suppose that the theorem holds for all values less than $k$. Let us prove that the theorem still holds if $t(S',v) = k$. We have that $v$ must have a neighbor $u$ such that $t(S',u) < t(S,u)$ and $t(S',u) < k$ because, otherwise, we would have that $t(S',v) = t(S,v)$. Thus, by our inductive hypothesis, there is a $S'$-infection path from some vertex $v'$ to $u$ and, since $t(S',u) < t(S',v) = k$, there is a $S'$-infection path from $v'$ to $v$.
\end{proof}

The next lemma, which is valid for every graph with maximum degree 3, is the main technical lemma of this section.
\begin{lemma}
\label{lemainfcam}
Let $G$ be a connected graph with $\Delta(G)=3$ and $k$ a non-negative integer. Then, $t(G)\geq k$ if and only if $G$ has an induced path $P$ where either all vertices of $V(P)$ have degree 3 and $|E(P)|\geq 2k-2$ or all vertices of $V(P)$ have degree 3, except for one of his extremities, which has degree 2, and $|E(P)|\geq k-1$.
\end{lemma}

\begin{proof}
First, suppose that $t(G)\geq k$. Let us prove that $G$ has an induced path $P$ where either all vertices in $V(P)$ have degree 3 and $|E(P)| \geq 2k-2$ or all vertices in $V(P)$ have degree 3, except for one of his extremities, which has degree 2, and $|E(P)| \geq k-1$.

Since $t(G) \geq k$, there is a percolating set $S$ such that $t(S) \geq k$. Let $t=t(S)$ and let $v$ be a vertex that is infected by $S$ at time $t$. Note that $v$ cannot have degree 1, because otherwise $v\in S$, a contradiction. So, let us divide the proof in 2 cases.

The first case occurs when $v$ has degree 2. In this case, let $P= v_1,\hdots,v_{t-1},v_t = v$ be a $S$-infection path where each $v_i$ is infected at time $i$ by $S$. Thus, we have that all vertices $v_1,v_2,\hdots,v_{t-2},v_{t-1}$ have degree three because each vertex $v_i$, for $1 \leq i \leq t-1$, must have two neighbors infected by $S$ at time $\leq i-1$ and, additionally, $v_i$ also has $v_{i+1}$, which is the next vertex in $P$, as his neighbor. Thus, since $\Delta = 3$, each vertex $v_i \in V(P)$, for $1 \leq i \leq t-1$, has exactly one neighbor infected at time $i+1$ by $S$, which is the vertex $v_{i+1}$, and has no neighbor infected at time $\geq i+2$ by $S$, which implies that no two consecutive vertices in $P$ are neighbors. Therefore, we have that $P$ is an induced path in $G$ such that $|E(P)| = t-1 \geq k-1$. 

The second case occurs when $v$ has degree 3. Thus, since there is no vertex infected by $S$ at time greater than the time of $v$, we have that $v$ has one neighbor infected by $S$ at time $\leq t-1$, another neighbor infected by $S$ at time $t-1$ and yet another neighbor infected by $S$ at time either $t-1$ or $t$. Let $t'$ be the infection time of the neighbor of $v'$ that is infected by $S$ at the greatest time among the neighbors of $v$, which may be $t$ or $t-1$. Let $P_1 = v_1,\hdots,v_{t-1},v_t = v$ and $P_2 = v'_1,\hdots,v'_{t'-1},v'_{t'}$ be two $S$-infection paths where each $v_i$ and $v'_i$ is infected at time $i$ by $S$ and $v'_{t'}$ is the neighbor of $v$ that is infected by $S$ at time $t'$. By the same arguments used in the first case, we have that both $P_1$ and $P_2$ are induced paths in $G$. Additionally, no vertex of $P_1$ is adjacent to a vertex of $P_2$, except the vertex $v$ that is adjacent to $v'_{t'}$ because, otherwise, we would have either that $v_{t-1} = v'_{t'}$ or that $v_t = v'_{t'}$.
Thus, let $P$ be the path $v_1,\hdots,v_{t-1},v, v'_{t'},\hdots,v'_2,v'_1$, which is the resulting path from the union of the induced paths $P_1$ and $P_2$ through the edge between the vertices $v$ and $v'_{t'}$. We have that $P$ is an induced path in $G$ because there is no edge between two non-consecutive vertices of $P$. We also have that all vertices in $P$ have degree 3 and $|E(P)| = t + t' - 1 \geq t + (t-1) - 1 \geq 2k-2$.

Now, suppose that $G$ has an induced path $P$ where either all of his vertices has degree 3 and $|E(P)| \geq 2k-2$ or all of his vertices has degree 3, except for one of his extremities, which has degree 2, and $|E(P)| \geq k-1$. Let us prove that $t(G) \geq k$.

Suppose that $G$ has an induced path $P = v_1,v_2,\hdots,v_t$, for $t \geq k$, where all of his vertices have degree 3, except for $v_t$, which has degree 2. Since each vertex $v_i$ has degree 3, except $v_t$, which has degree 2, and $P$ is an induced path, then each vertex $v_i$ has exactly one neighbor outside of $P$, except for $v_1$, which has two neighbors outside of $P$. Let $S'$ be the set of neighbors of each $v_i$ that is not in $V(P)$. It is easy to see that $S'$ infects all vertices in $P$. Since $v_t$ has only two neighbors and one of them is in $S'$, then $t(S',v_t) = t(S',v_{t-1}) + 1$. Since $v_{t-1}$ has 3 neighbors, where one of them is in $S'$ and the other is $v_t$, which is infected by $S'$ after $v_{t-1}$, then $t(S',v_{t-1}) = t(S',v_{t-2}) + 1$. Therefore, by this same argument, we have that, for all $2 \leq i \leq t$, we have that $t(S',v_i) = t(S',v_{i-1}) + 1$. Since $v_1$ has two neighbors in $S'$, then $v_1$ is infected by $S'$ at time 1 and, hence, for all $1 \leq i \leq t$, $t(S',v_i)=i$. Thus, there is a set $S'$ that infects $v_t$ at time $t$.

Now, suppose that $G$ has an induced path $P = v_1,v_2,\hdots,v_{2t-2},v_{2t-1}$, where all of his vertices have degree 3. Since each vertex $v_i$ has degree 3 and $P$ is an induced path, then each vertex $v_i$ has exactly one neighbor outside of $P$, except for both $v_1$ and $v_{2t-1}$, which have two neighbors outside of $P$. Let $S'$ be the set of neighbors of each $v_i$ that is not in $V(P)$. Again, it is easy to see that $S'$ infects all vertices in $P$. Also, similarly to the prior case, it is not hard to see that all vertices at distance $d$ of $v_t$ are infected by $S'$ at time $t - d$, i.e., $S'$ infects a vertex $v_i$ at time $t-|t-i|$.

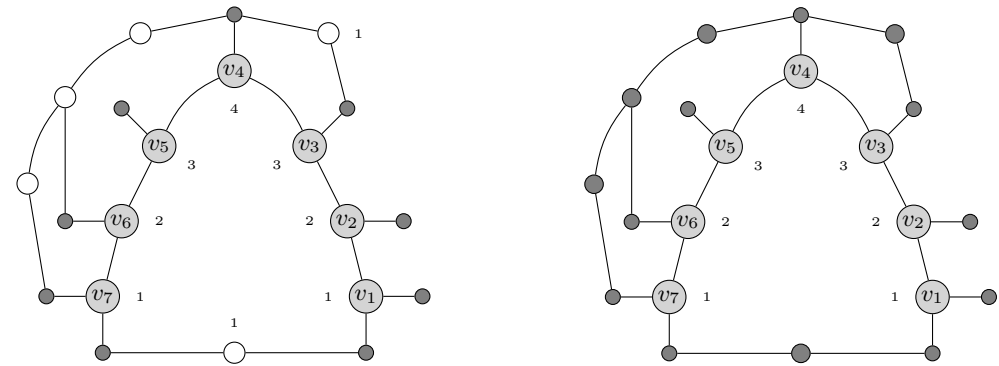
\begin{figure}[ht]
\centering
\begin{minipage}[c]{.45\textwidth}
\centering
\begin{tikzpicture}[scale=1]
\tikzstyle{vertex}=[draw,circle,fill=black!17,minimum size=10pt,inner sep=1pt]
\tikzstyle{nvertex}=[draw,circle,fill=black!50,minimum size=2pt,inner sep=2pt]
\tikzstyle{overtex}=[draw,circle,fill=white,minimum size=8pt,inner sep=2pt]

\node[vertex] (v1) at (1.75,-3) {$v_1$};
\node[vertex] (v2) at (1.5,-2) {$v_2$};
\node[vertex] (v3) at (1,-1) {$v_3$};
\node[vertex] (v4) at (0,0) {$v_4$};
\node[vertex] (v5) at (-1,-1) {$v_5$};
\node[vertex] (v6) at (-1.5,-2) {$v_6$};
\node[vertex] (v7) at (-1.75,-3) {$v_7$};

\node[shift={(180:0.5cm)}] (v1l) at (v1) {\tiny{1}};
\node[shift={(180:0.5cm)}] (v2l) at (v2) {\tiny{2}};
\node[shift={(210:0.5cm)}] (v3l) at (v3) {\tiny{3}};
\node[shift={(270:0.5cm)}] (v4l) at (v4) {\tiny{4}};
\node[shift={(330:0.5cm)}] (v5l) at (v5) {\tiny{3}};
\node[shift={(0:0.5cm)}] (v6l) at (v6) {\tiny{2}};
\node[shift={(0:0.5cm)}] (v7l) at (v7) {\tiny{1}};

\node[nvertex] (v1n1) at (1.75,-3.75) {};
\node[nvertex] (v1n2) at (2.5,-3) {};
\node[nvertex] (v2n) at (2.25,-2) {};
\node[nvertex] (v3n) at (1.5,-0.5) {};
\node[nvertex] (v4n) at (0,0.75) {};
\node[nvertex] (v5n) at (-1.5,-0.5) {};
\node[nvertex] (v6n) at (-2.25,-2) {};
\node[nvertex] (v7n1) at (-2.5,-3) {};
\node[nvertex] (v7n2) at (-1.75,-3.75) {};

\node[overtex] (vo1) at (-1.25,0.5) {};
\node[overtex] (vo2) at (1.25,0.5) {};
\node[overtex] (vo3) at (-2.75,-1.5) {};
\node[overtex] (vo4) at (0,-3.75) {};
\node[overtex] (vo5) at (-2.25,-0.35) {};

\node[shift={(0:0.4cm)}] (v6l) at (vo2) {\tiny{1}};
\node[shift={(90:0.4cm)}] (v7l) at (vo4) {\tiny{1}};

\draw[-] (v1) to (v2);
\draw[-] (v2) to (v3);
\draw[-] (v3) to [bend right=20] (v4);
\draw[-] (v4) to [bend right=20] (v5);
\draw[-] (v5) to (v6);
\draw[-] (v6) to (v7);

\draw[-] (v1) to (v1n1);
\draw[-] (v1) to (v1n2);
\draw[-] (v2) to (v2n);
\draw[-] (v3) to (v3n);
\draw[-] (v4) to (v4n);
\draw[-] (v5) to (v5n);
\draw[-] (v6) to (v6n);
\draw[-] (v7) to (v7n1);
\draw[-] (v7) to (v7n2);

\draw[-] (vo1) to (v4n);
\draw[-] (vo1) to [bend right=15] (vo5);
\draw[-] (vo2) to (v3n);
\draw[-] (vo2) to (v4n);
\draw[-] (vo3) to (v7n1);
\draw[-] (vo4) to (v1n1);
\draw[-] (vo4) to (v7n2);
\draw[-] (vo5) to [bend right=15] (vo3);
\draw[-] (vo5) to (v6n);

\end{tikzpicture}
\end{minipage} 
%\hfill
\begin{minipage}[c]{.45\textwidth}
\centering
\begin{tikzpicture}[scale=1]
\tikzstyle{vertex}=[draw,circle,fill=black!17,minimum size=10pt,inner sep=1pt]
\tikzstyle{nvertex}=[draw,circle,fill=black!50,minimum size=2pt,inner sep=2pt]
\tikzstyle{overtex}=[draw,circle,fill=black!50,minimum size=7pt,inner sep=2pt]

\node[vertex] (v1) at (1.75,-3) {$v_1$};
\node[vertex] (v2) at (1.5,-2) {$v_2$};
\node[vertex] (v3) at (1,-1) {$v_3$};
\node[vertex] (v4) at (0,0) {$v_4$};
\node[vertex] (v5) at (-1,-1) {$v_5$};
\node[vertex] (v6) at (-1.5,-2) {$v_6$};
\node[vertex] (v7) at (-1.75,-3) {$v_7$};

\node[shift={(180:0.5cm)}] (v1l) at (v1) {\tiny{1}};
\node[shift={(180:0.5cm)}] (v2l) at (v2) {\tiny{2}};
\node[shift={(210:0.5cm)}] (v3l) at (v3) {\tiny{3}};
\node[shift={(270:0.5cm)}] (v4l) at (v4) {\tiny{4}};
\node[shift={(330:0.5cm)}] (v5l) at (v5) {\tiny{3}};
\node[shift={(0:0.5cm)}] (v6l) at (v6) {\tiny{2}};
\node[shift={(0:0.5cm)}] (v7l) at (v7) {\tiny{1}};

\node[nvertex] (v1n1) at (1.75,-3.75) {};
\node[nvertex] (v1n2) at (2.5,-3) {};
\node[nvertex] (v2n) at (2.25,-2) {};
\node[nvertex] (v3n) at (1.5,-0.5) {};
\node[nvertex] (v4n) at (0,0.75) {};
\node[nvertex] (v5n) at (-1.5,-0.5) {};
\node[nvertex] (v6n) at (-2.25,-2) {};
\node[nvertex] (v7n1) at (-2.5,-3) {};
\node[nvertex] (v7n2) at (-1.75,-3.75) {};

\node[overtex] (vo1) at (-1.25,0.5) {};
\node[overtex] (vo2) at (1.25,0.5) {};
\node[overtex] (vo3) at (-2.75,-1.5) {};
\node[overtex] (vo4) at (0,-3.75) {};
\node[overtex] (vo5) at (-2.25,-0.35) {};

\draw[-] (v1) to (v2);
\draw[-] (v2) to (v3);
\draw[-] (v3) to [bend right=20] (v4);
\draw[-] (v4) to [bend right=20] (v5);
\draw[-] (v5) to (v6);
\draw[-] (v6) to (v7);

\draw[-] (v1) to (v1n1);
\draw[-] (v1) to (v1n2);
\draw[-] (v2) to (v2n);
\draw[-] (v3) to (v3n);
\draw[-] (v4) to (v4n);
\draw[-] (v5) to (v5n);
\draw[-] (v6) to (v6n);
\draw[-] (v7) to (v7n1);
\draw[-] (v7) to (v7n2);

\draw[-] (vo1) to (v4n);
\draw[-] (vo1) to [bend right=15] (vo5);
\draw[-] (vo2) to (v3n);
\draw[-] (vo2) to (v4n);
\draw[-] (vo3) to (v7n1);
\draw[-] (vo4) to (v1n1);
\draw[-] (vo4) to (v7n2);
\draw[-] (vo5) to [bend right=15] (vo3);
\draw[-] (vo5) to (v6n);

\end{tikzpicture}
\end{minipage}

\caption{A graph with $\Delta = 3$ infected by the set $S'$ to the left and by the percolating set $S$ to the right.}
\label{infcammesmo}
\end{figure}

Thus, in both cases, we have that $S'$ infects each vertex $v_i$ in $P$ at time $t-|t-i|$ (in the first case, since $1 \leq i \leq t$, we have that $i = t - |t-i|$). Let $Y$ be the set of vertices that are neither in $V(P)$ nor in the neighborhood of any vertex in $V(P)$. Let $S = S' \cup Y$. We have that $S$ is a percolating set because all vertices in $V(G)$ are either in $V(P)$ or in the neighborhood of some vertex in $V(P)$ or in $Y$, and $S'$ infects all vertices that are either in $V(P)$ or in the neighborhood of some vertex in $V(P)$. Also, since all neighbors of the vertices in $V(P)$ are in $S'$, $S$ cannot possible infect the vertices in $V(P)$ in a different time than $S'$, as exemplified in Figure \ref{infcammesmo}. Thus, since $S$ is a percolating set that infects $v_t$ at time $t$, we have that $t(S) \geq t \geq k$ and, hence, $t(G) \geq k$.

\end{proof}

% ----------------------------------------------------------------------------------
% ----------------------------------------------------------------------------------
% ----------------------------------------------------------------------------------
% ----------------------------------------------------------------------------------
Before proving the NP-completeness result of this section, we use Lemma \ref{lemainfcam} to show that the Percolation Time Problem is polynomial time solvable for $k=O(\log n)$ when the graph has maximum degree 3.

\begin{theorem}\label{teo-d3-logn}
If $G$ is a graph with maximum degree 3, then deciding whether $t(G)\geq k$ can be done in time $O(n^{2c+3})$ for $k\leq c\cdot\log_2 n$ and $c>0$.
\end{theorem}

\begin{proof}[sketch of the proof]
We can decide whether $t(G) \geq k$ by making use of a modified version of the depth-first search. This version of the depth-first search with maximum search depth $\ell$ traverses all paths with $\ell+1$ vertices starting from some vertex $v$. For each $v \in V(G)$, we will run this version of the depth-first search starting in $v$. If $d(v) = 2$, we run the modified depth-first search with maximum search depth $k-1$. If $d(v) = 3$, we run the modified depth-first search with maximum search depth $2k-2$. If there is a vertex $v$ such that the depth-first search that starts in $v$ finds a path that is an induced path, reaches the maximum depth and passes only by vertices of degree 3, except maybe for $v$, then, by Lemma \ref{lemainfcam}, $t(G) \geq k$. Otherwise, $t(G) < k$.

Now, let us show that this algorithm runs in polynomial time. For each vertex $v$ in $G$, there are at most $3\cdot 2^{\ell-2}$ paths of length $\ell$ in $G$ that starts in $v$, for any $\ell$. In this case, since $\ell\leq 2k-2$, there are at most $3\cdot 2^{2k-2}=3n^{2c}/4$ paths of length $\ell$ in $G$ that starts in $v$. Therefore, since we have $n$ vertices  and we take time $O(n^2)$ to obtain each path, the algorithm runs in time $O(n^{2c+3})$.
\end{proof}

% ----------------------------------------------------------------------------------
% ----------------------------------------------------------------------------------
% ----------------------------------------------------------------------------------
% ----------------------------------------------------------------------------------
Thus, if $k=O(\log n)$, we can find whether $t(G) \geq k$ in polynomial time for every graph $G$ with $\Delta(G) = 3$. However, the following theorem states that the Percolation Time Problem is NP-complete, even when $G$ is restricted to be a grid graph with $\Delta = 3$.

\begin{theorem}
\label{teogrid}
Deciding whether $t(G)\geq k$ is NP-complete when the input $G$ is restricted to be a grid graph with $\Delta(G)\leq 3$.
\end{theorem}

\begin{proof}
Clearly, the problem is in NP. To prove that the problem is also NP-hard, we obtained a reduction from the Longest Path problem with input restricted to be grid graphs with maximum degree 3. The Longest Path problem with input restricted to be grid graphs with maximum degree 3 is a NP-complete problem because the Hamiltonian Path Problem with input restricted to be grid graphs with maximum degree 3 is also NP-complete \cite{papa1} and there is a trivial reduction from the Hamiltonian Path Problem to the Longest Path problem that does not change the input graph: $G$ has an Hamiltonian Path if and only if $G$ has a path greater or equal to $n-1$.

\begin{figure}[ht]
\centering
\begin{tikzpicture}[scale=1]
\tikzstyle{vertex}=[draw,circle,fill=black!25,minimum size=10pt,inner sep=1pt]

\def \scale {1.5}

\node[vertex] (v1) at (0*\scale,0*\scale) {};
\node[vertex] (v2) at (1*\scale,0*\scale) {};
\node[vertex] (v3) at (2*\scale,0*\scale) {};
\node[vertex] (v4) at (3*\scale,0*\scale) {};

\node[vertex] (v5) at (0*\scale,1*\scale) {};
\node[vertex] (v6) at (1*\scale,1*\scale) {};
\node[vertex] (v7) at (2*\scale,1*\scale) {};
\node[vertex] (v8) at (3*\scale,1*\scale) {};

\node[vertex] (v9) at (0*\scale,2*\scale) {};
\node[vertex] (v10) at (3*\scale,2*\scale) {};

\node[vertex] (v11) at (0*\scale,3*\scale) {};
\node[vertex] (v12) at (1*\scale,3*\scale) {};
\node[vertex] (v13) at (2*\scale,3*\scale) {};
\node[vertex] (v14) at (3*\scale,3*\scale) {};

\node[vertex] (v15) at (0*\scale,4*\scale) {};
\node[vertex] (v16) at (1*\scale,4*\scale) {};
\node[vertex] (v17) at (2*\scale,4*\scale) {};
\node[vertex] (v18) at (3*\scale,4*\scale) {};
\node[vertex] (v19) at (-1*\scale,2*\scale) {};

\draw[-] (v1) to (v2);
\draw[-] (v2) to (v3);
\draw[-] (v3) to (v4);
\draw[-] (v5) to (v6);
\draw[-] (v6) to (v7);
\draw[-] (v7) to (v8);

\draw[-] (v1) to (v5);
\draw[-] (v2) to (v6);
\draw[-] (v3) to (v7);
\draw[-] (v4) to (v8);

\draw[-] (v5) to (v9);
\draw[-] (v9) to (v11);
\draw[-] (v8) to (v10);
\draw[-] (v10) to (v14);

\draw[-] (v11) to (v12);
\draw[-] (v12) to (v13);
\draw[-] (v13) to (v14);
\draw[-] (v15) to (v16);
\draw[-] (v16) to (v17);
\draw[-] (v17) to (v18);

\draw[-] (v11) to (v15);
\draw[-] (v12) to (v16);
\draw[-] (v13) to (v17);
\draw[-] (v14) to (v18);

\draw[-] (v9) to (v19);

\end{tikzpicture}
\caption{\label{exemplo}Grid graph with $\Delta = 3$}
\end{figure}
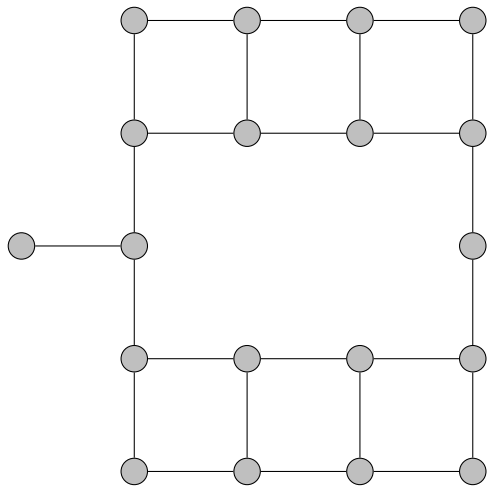

Consider the following reduction from the Longest Path Problem's instance $(G,k)$ where $G$ is restricted to be a grid graph with maximum degree 3 to the Percolation Time Problem's instance $(G',3k+2)$ where $G'$ is also a grid graph with maximum degree 3: Multiply the scale of the grid $G$ by three. Each edge in $G$ becomes a path in $G'$ with 4 vertices where the vertices at the extremities are vertices that were originally in $G$. Let us call an \textit{original vertex} the vertices in $G'$ that were originally in $G$. After that, for each original vertex $v$, if $d(v) < 3$, add to $G'$ $3 - d(v)$ vertices in any free position in the grid adjacent to $v$ and link them to $v$. Thus, after we do that, each original vertex has degree 3 in $G'$. Henceforth, if a vertex in $G'$ is not an original vertex at this point, then we will call it an \textit{auxiliary vertex}. Note that each auxiliary vertex is adjacent to exactly one original vertex and each original vertex is adjacent to 3 auxiliary vertices.

After that, for each auxiliary vertex $v$, add a new vertex adjacent to $v$ in the following manner: if the original neighbor of $v$ is located above it, add a vertex adjacent to $v$ at his left position, if there is not one there already, and link it to $v$. If the original neighbor of $v$ is located below it, add a vertex adjacent to $v$ at his right position, if there is not one there already, and link it to $v$. If the original neighbor of $v$ is located at his left position, add a vertex adjacent to $v$ at the position below it, if there is not one there already, and link it to $v$. If the original neighbor of $v$ is located at his right position, add a vertex adjacent to $v$ at the position above it, if there is not one there already, and link it to $v$. The Figure \ref{bloco} show how a 4x4 block will look like in $G'$ before and after we add these vertices.

\begin{figure}[ht]
\centering
\begin{minipage}[c]{.30\textwidth}
\centering
\begin{tikzpicture}[scale=1]
\tikzstyle{vertex}=[draw,circle,fill=black!25,minimum size=10pt,inner sep=1pt]
\tikzstyle{intvertex}=[draw,circle,fill=black!50,minimum size=2pt,inner sep=2pt]

\node[vertex] (ce) at (0,2) {};
\node[vertex] (cd) at (2,2) {};
\node[vertex] (be) at (0,0) {};
\node[vertex] (bd) at (2,0) {};

\node[intvertex] (ce1cd) at (0.67,2) {};
\node[intvertex] (ce2cd) at (1.33,2) {};
\node[intvertex] (cd1bd) at (2,1.33) {};
\node[intvertex] (cd2bd) at (2,0.67) {};
\node[intvertex] (bd1be) at (1.33,0) {};
\node[intvertex] (bd2be) at (0.67,0) {};
\node[intvertex] (be1ce) at (0,0.67) {};
\node[intvertex] (be2ce) at (0,1.33) {};

\draw[-] (ce) to (ce1cd);
\draw[-] (ce1cd) to (ce2cd);
\draw[-] (ce2cd) to (cd);
\draw[-] (cd) to (cd1bd);
\draw[-] (cd1bd) to (cd2bd);
\draw[-] (cd2bd) to (bd);
\draw[-] (bd) to (bd1be);
\draw[-] (bd1be) to (bd2be);
\draw[-] (bd2be) to (be);
\draw[-] (be) to (be1ce);
\draw[-] (be1ce) to (be2ce);
\draw[-] (be2ce) to (ce);

\end{tikzpicture}
\end{minipage} 
%\hfill
\begin{minipage}[c]{.30\textwidth}
\centering
\begin{tikzpicture}[scale=1]
\tikzstyle{vertex}=[draw,circle,fill=black!25,minimum size=10pt,inner sep=1pt]
\tikzstyle{intvertex}=[draw,circle,fill=black!50,minimum size=2pt,inner sep=2pt]

\def \scale {0.667}

\node[vertex] (ce) at (0*\scale,3*\scale) {};
\node[vertex] (cd) at (3*\scale,3*\scale) {};
\node[vertex] (be) at (0*\scale,0*\scale) {};
\node[vertex] (bd) at (3*\scale,0*\scale) {};

\node[intvertex] (ce1cd) at (1*\scale,3*\scale) {};
\node[intvertex] (ce2cd) at (2*\scale,3*\scale) {};
\node[intvertex] (cd1bd) at (3*\scale,2*\scale) {};
\node[intvertex] (cd2bd) at (3*\scale,1*\scale) {};
\node[intvertex] (bd1be) at (2*\scale,0*\scale) {};
\node[intvertex] (bd2be) at (1*\scale,0*\scale) {};
\node[intvertex] (be1ce) at (0*\scale,1*\scale) {};
\node[intvertex] (be2ce) at (0*\scale,2*\scale) {};

\node[intvertex] (c1) at (1*\scale,2*\scale) {};
\node[intvertex] (c2) at (2*\scale,1*\scale) {};

\node[intvertex] (ce2cdf) at (2*\scale,4*\scale) {};
\node[intvertex] (cd1bdf) at (4*\scale,2*\scale) {};
\node[intvertex] (bd2bef) at (1*\scale,-1*\scale) {};
\node[intvertex] (be1cef) at (-1*\scale,1*\scale) {};

\draw[-] (ce) to (ce1cd);
\draw[-] (ce1cd) to (ce2cd);
\draw[-] (ce2cd) to (cd);
\draw[-] (cd) to (cd1bd);
\draw[-] (cd1bd) to (cd2bd);
\draw[-] (cd2bd) to (bd);
\draw[-] (bd) to (bd1be);
\draw[-] (bd1be) to (bd2be);
\draw[-] (bd2be) to (be);
\draw[-] (be) to (be1ce);
\draw[-] (be1ce) to (be2ce);
\draw[-] (be2ce) to (ce);

\draw[-] (c1) to (ce1cd);
\draw[-] (c1) to (be2ce);
\draw[-] (c2) to (cd2bd);
\draw[-] (c2) to (bd1be);

\draw[-] (ce2cdf) to (ce2cd);
\draw[-] (cd1bdf) to (cd1bd);
\draw[-] (bd2bef) to (bd2be);
\draw[-] (be1cef) to (be1ce);

\end{tikzpicture}
\end{minipage}

\caption{4x4 block before and after addition of the auxiliary vertices' neighbors}
\label{bloco}
\end{figure}
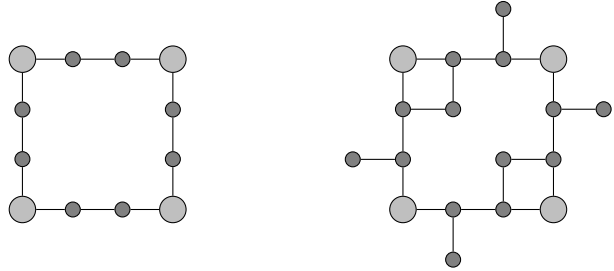

Then, for each auxiliary vertex $v$, if $d(v) = 2$, add a new vertex adjacent to $v$ in the following position: if the original neighbor of $v$ is at the left position of $v$, add a vertex adjacent to $v$ at his right position. If the original neighbor of $v$ is at the right position of $v$, add a vertex adjacent to $v$ at his left position. If the original neighbor of $v$ is below $v$, add a vertex adjacent to $v$ above $v$. If the original neighbor of $v$ is above $v$, add a vertex adjacent to $v$ below $v$. 

Thus, the construction of $G'$ is finished. Since $G$ is a grid graph and, every time an original vertex and an auxiliary vertex are in adjacent positions in the grid, they are linked, then $G'$ is a grid graph.

Note that all original and auxiliary vertices have degree 3 and they are the only vertices that have degree 3. Let us call \textit{corner vertex} all the vertices that have degree 2 in $G'$. Also, note that, for each corner vertex, there is exactly one original vertex at distance 2 of it, and, for each original vertex, there is exactly one corner vertex at distance 2 of it. This happens because each original vertex has degree exactly three. Let $f$ be the bijective function that maps each original vertex to the corner vertex that is at distance 2 of it. The Figure \ref{exemplototal} shows the reduction applied to the grid graph of the Figure \ref{exemplo}. It is worth noting that, in $G'$, a path $P$ that has only original and auxiliary vertices and starts with an original vertex, it has length multiple of 3 if and only if it ends in an original vertex. Also, for each 3 consecutive vertices of this path, two are auxiliary vertices and one is an original vertex.

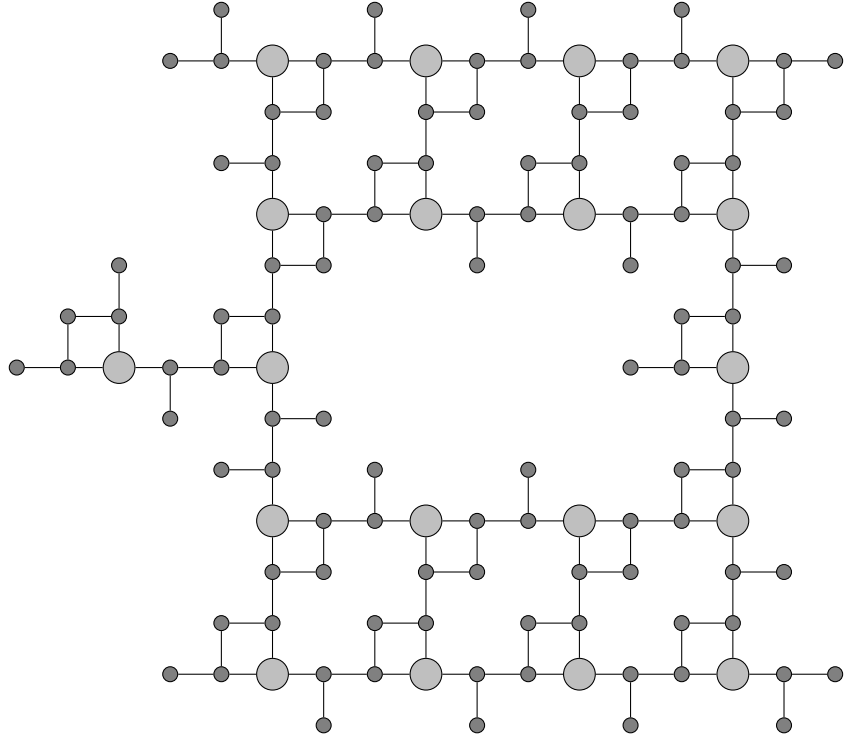
\begin{figure}[ht]
\centering
\begin{tikzpicture}[scale=0.85]
\tikzstyle{vertex}=[draw,circle,fill=black!25,minimum size=12pt,inner sep=1pt]
\tikzstyle{intvertex}=[draw,circle,fill=black!50,minimum size=2pt,inner sep=2pt]

\def \scale {2.4}
\def \scalei {0.8}

\node[vertex] (v1) at (0*\scale,0*\scale) {};
\node[vertex] (v2) at (1*\scale,0*\scale) {};
\node[vertex] (v3) at (2*\scale,0*\scale) {};
\node[vertex] (v4) at (3*\scale,0*\scale) {};

\node[vertex] (v5) at (0*\scale,1*\scale) {};
\node[vertex] (v6) at (1*\scale,1*\scale) {};
\node[vertex] (v7) at (2*\scale,1*\scale) {};
\node[vertex] (v8) at (3*\scale,1*\scale) {};

\node[vertex] (v9) at (0*\scale,2*\scale) {};
\node[vertex] (v10) at (3*\scale,2*\scale) {};

\node[vertex] (v11) at (0*\scale,3*\scale) {};
\node[vertex] (v12) at (1*\scale,3*\scale) {};
\node[vertex] (v13) at (2*\scale,3*\scale) {};
\node[vertex] (v14) at (3*\scale,3*\scale) {};

\node[vertex] (v15) at (0*\scale,4*\scale) {};
\node[vertex] (v16) at (1*\scale,4*\scale) {};
\node[vertex] (v17) at (2*\scale,4*\scale) {};
\node[vertex] (v18) at (3*\scale,4*\scale) {};
\node[vertex] (v19) at (-1*\scale,2*\scale) {};

\node[intvertex] (vi1) at (1*\scalei,0*\scalei) {};
\node[intvertex] (vi2) at (2*\scalei,0*\scalei) {};
\node[intvertex] (vi3) at (4*\scalei,0*\scalei) {};
\node[intvertex] (vi4) at (5*\scalei,0*\scalei) {};
\node[intvertex] (vi5) at (7*\scalei,0*\scalei) {};
\node[intvertex] (vi6) at (8*\scalei,0*\scalei) {};
\node[intvertex] (vi7) at (1*\scalei,3*\scalei) {};
\node[intvertex] (vi8) at (2*\scalei,3*\scalei) {};
\node[intvertex] (vi9) at (4*\scalei,3*\scalei) {};
\node[intvertex] (vi10) at (5*\scalei,3*\scalei) {};
\node[intvertex] (vi11) at (7*\scalei,3*\scalei) {};
\node[intvertex] (vi12) at (8*\scalei,3*\scalei) {};
\node[intvertex] (vi13) at (0*\scalei,1*\scalei) {};
\node[intvertex] (vi14) at (0*\scalei,2*\scalei) {};
\node[intvertex] (vi15) at (3*\scalei,1*\scalei) {};
\node[intvertex] (vi16) at (3*\scalei,2*\scalei) {};
\node[intvertex] (vi17) at (6*\scalei,1*\scalei) {};
\node[intvertex] (vi18) at (6*\scalei,2*\scalei) {};
\node[intvertex] (vi19) at (9*\scalei,1*\scalei) {};
\node[intvertex] (vi20) at (9*\scalei,2*\scalei) {};

\node[intvertex] (vi21) at (0*\scalei,4*\scalei) {};
\node[intvertex] (vi22) at (0*\scalei,5*\scalei) {};
\node[intvertex] (vi23) at (9*\scalei,4*\scalei) {};
\node[intvertex] (vi24) at (9*\scalei,5*\scalei) {};
\node[intvertex] (vi25) at (0*\scalei,7*\scalei) {};
\node[intvertex] (vi26) at (0*\scalei,8*\scalei) {};
\node[intvertex] (vi27) at (9*\scalei,7*\scalei) {};
\node[intvertex] (vi28) at (9*\scalei,8*\scalei) {};

\node[intvertex] (vi29) at (1*\scalei,9*\scalei) {};
\node[intvertex] (vi30) at (2*\scalei,9*\scalei) {};
\node[intvertex] (vi31) at (4*\scalei,9*\scalei) {};
\node[intvertex] (vi32) at (5*\scalei,9*\scalei) {};
\node[intvertex] (vi33) at (7*\scalei,9*\scalei) {};
\node[intvertex] (vi34) at (8*\scalei,9*\scalei) {};
\node[intvertex] (vi35) at (1*\scalei,12*\scalei) {};
\node[intvertex] (vi36) at (2*\scalei,12*\scalei) {};
\node[intvertex] (vi37) at (4*\scalei,12*\scalei) {};
\node[intvertex] (vi38) at (5*\scalei,12*\scalei) {};
\node[intvertex] (vi39) at (7*\scalei,12*\scalei) {};
\node[intvertex] (vi40) at (8*\scalei,12*\scalei) {};
\node[intvertex] (vi41) at (0*\scalei,10*\scalei) {};
\node[intvertex] (vi42) at (0*\scalei,11*\scalei) {};
\node[intvertex] (vi43) at (3*\scalei,10*\scalei) {};
\node[intvertex] (vi44) at (3*\scalei,11*\scalei) {};
\node[intvertex] (vi45) at (6*\scalei,10*\scalei) {};
\node[intvertex] (vi46) at (6*\scalei,11*\scalei) {};
\node[intvertex] (vi47) at (9*\scalei,10*\scalei) {};
\node[intvertex] (vi48) at (9*\scalei,11*\scalei) {};

\node[intvertex] (vi49) at (-1*\scalei,6*\scalei) {};
\node[intvertex] (vi50) at (-2*\scalei,6*\scalei) {};

\node[intvertex] (vix13) at (1*\scalei,-1*\scalei) {};
\node[intvertex] (vx1) at (-1*\scalei,0*\scalei) {};
\node[intvertex] (vxx1) at (-2*\scalei,0*\scalei) {};
\node[intvertex] (vc2) at (2*\scalei,1*\scalei) {};
\node[intvertex] (vc3) at (5*\scalei,1*\scalei) {};
\node[intvertex] (vc4) at (8*\scalei,1*\scalei) {};
\node[intvertex] (vx4) at (10*\scalei,0*\scalei) {};
\node[intvertex] (vxx4) at (11*\scalei,0*\scalei) {};
\node[intvertex] (vc5) at (1*\scalei,2*\scalei) {};
\node[intvertex] (vc6) at (4*\scalei,2*\scalei) {};
\node[intvertex] (vc7) at (7*\scalei,2*\scalei) {};
\node[intvertex] (vc8) at (8*\scalei,4*\scalei) {};

\node[intvertex] (vix3) at (4*\scalei,-1*\scalei) {};
\node[intvertex] (vix5) at (7*\scalei,-1*\scalei) {};
\node[intvertex] (vc1) at (-1*\scalei,1*\scalei) {};
\node[intvertex] (vix20) at (10*\scalei,2*\scalei) {};
\node[intvertex] (vix8) at (2*\scalei,4*\scalei) {};
\node[intvertex] (vix10) at (5*\scalei,4*\scalei) {};
\node[intvertex] (vix21) at (-1*\scalei,4*\scalei) {};
\node[intvertex] (vix24) at (10*\scalei,5*\scalei) {};
\node[intvertex] (vc9) at (-1*\scalei,7*\scalei) {};
\node[intvertex] (vix28) at (10*\scalei,8*\scalei) {};

\node[intvertex] (vix22) at (1*\scalei,5*\scalei) {};
%\node[intvertex] (vxx9) at (2*\scalei,6*\scalei) {};
\node[intvertex] (vc10) at (8*\scalei,7*\scalei) {};
\node[intvertex] (vx10) at (8*\scalei,6*\scalei) {};
\node[intvertex] (vxx10) at (7*\scalei,6*\scalei) {};

\node[intvertex] (vc11) at (1*\scalei,8*\scalei) {};
\node[intvertex] (vc12) at (2*\scalei,10*\scalei) {};
\node[intvertex] (vc13) at (5*\scalei,10*\scalei) {};
\node[intvertex] (vc14) at (8*\scalei,10*\scalei) {};
\node[intvertex] (vc15) at (1*\scalei,11*\scalei) {};
\node[intvertex] (vx15) at (-1*\scalei,12*\scalei) {};
\node[intvertex] (vxx15) at (-2*\scalei,12*\scalei) {};
\node[intvertex] (vc16) at (4*\scalei,11*\scalei) {};
\node[intvertex] (vc17) at (7*\scalei,11*\scalei) {};
\node[intvertex] (vix48) at (8*\scalei,13*\scalei) {};
\node[intvertex] (vx18) at (10*\scalei,12*\scalei) {};
\node[intvertex] (vxx18) at (11*\scalei,12*\scalei) {};
\node[intvertex] (vc19) at (-4*\scalei,7*\scalei) {};
\node[intvertex] (vx191) at (-4*\scalei,6*\scalei) {};
\node[intvertex] (vxx191) at (-5*\scalei,6*\scalei) {};
\node[intvertex] (vx192) at (-3*\scalei,7*\scalei) {};
\node[intvertex] (vxx192) at (-3*\scalei,8*\scalei) {};

\node[intvertex] (vix31) at (4*\scalei,8*\scalei) {};
\node[intvertex] (vix33) at (7*\scalei,8*\scalei) {};
\node[intvertex] (vix36) at (2*\scalei,13*\scalei) {};
\node[intvertex] (vix38) at (5*\scalei,13*\scalei) {};
\node[intvertex] (vc18) at (10*\scalei,11*\scalei) {};
\node[intvertex] (vix41) at (-1*\scalei,10*\scalei) {};

\node[intvertex] (vix50) at (-2*\scalei,5*\scalei) {};
\node[intvertex] (vix51) at (10*\scalei,-1*\scalei) {};
\node[intvertex] (vix52) at (-1*\scalei,13*\scalei) {};

\draw[-] (v1) to (vi1);
\draw[-] (vi1) to (vi2);
\draw[-] (vi2) to (v2);
\draw[-] (v2) to (vi3);
\draw[-] (vi3) to (vi4);
\draw[-] (vi4) to (v3);
\draw[-] (v3) to (vi5);
\draw[-] (vi5) to (vi6);
\draw[-] (vi6) to (v4);

\draw[-] (v5) to (vi7);
\draw[-] (vi7) to (vi8);
\draw[-] (vi8) to (v6);
\draw[-] (v6) to (vi9);
\draw[-] (vi9) to (vi10);
\draw[-] (vi10) to (v7);
\draw[-] (v7) to (vi11);
\draw[-] (vi11) to (vi12);
\draw[-] (vi12) to (v8);

\draw[-] (v1) to (vi13);
\draw[-] (vi13) to (vi14);
\draw[-] (vi14) to (v5);
\draw[-] (v2) to (vi15);
\draw[-] (vi15) to (vi16);
\draw[-] (vi16) to (v6);
\draw[-] (v3) to (vi17);
\draw[-] (vi17) to (vi18);
\draw[-] (vi18) to (v7);
\draw[-] (v4) to (vi19);
\draw[-] (vi19) to (vi20);
\draw[-] (vi20) to (v8);

\draw[-] (v5) to (vi21);
\draw[-] (vi21) to (vi22);
\draw[-] (vi22) to (v9);
\draw[-] (v8) to (vi23);
\draw[-] (vi23) to (vi24);
\draw[-] (vi24) to (v10);
\draw[-] (v9) to (vi25);
\draw[-] (vi25) to (vi26);
\draw[-] (vi26) to (v11);
\draw[-] (v10) to (vi27);
\draw[-] (vi27) to (vi28);
\draw[-] (vi28) to (v14);

\draw[-] (v11) to (vi29);
\draw[-] (vi29) to (vi30);
\draw[-] (vi30) to (v12);
\draw[-] (v12) to (vi31);
\draw[-] (vi31) to (vi32);
\draw[-] (vi32) to (v13);
\draw[-] (v13) to (vi33);
\draw[-] (vi33) to (vi34);
\draw[-] (vi34) to (v14);

\draw[-] (v15) to (vi35);
\draw[-] (vi35) to (vi36);
\draw[-] (vi36) to (v16);
\draw[-] (v16) to (vi37);
\draw[-] (vi37) to (vi38);
\draw[-] (vi38) to (v17);
\draw[-] (v17) to (vi39);
\draw[-] (vi39) to (vi40);
\draw[-] (vi40) to (v18);

\draw[-] (v11) to (vi41);
\draw[-] (vi41) to (vi42);
\draw[-] (vi42) to (v15);
\draw[-] (v12) to (vi43);
\draw[-] (vi43) to (vi44);
\draw[-] (vi44) to (v16);
\draw[-] (v13) to (vi45);
\draw[-] (vi45) to (vi46);
\draw[-] (vi46) to (v17);
\draw[-] (v14) to (vi47);
\draw[-] (vi47) to (vi48);
\draw[-] (vi48) to (v18);

\draw[-] (v9) to (vi49);
\draw[-] (vi49) to (vi50);
\draw[-] (vi50) to (v19);

\draw[-] (v1) to (vx1);
\draw[-] (v4) to (vx4);
%\draw[-] (v9) to (vx9);
\draw[-] (v10) to (vx10);
\draw[-] (v15) to (vx15);
\draw[-] (v18) to (vx18);
\draw[-] (v19) to (vx191);
\draw[-] (v19) to (vx192);

\draw[-] (vx1) to (vc1);
\draw[-] (vx1) to (vxx1);
\draw[-] (vi1) to (vix13);
\draw[-] (vi2) to (vc2);
\draw[-] (vi3) to (vix3);
\draw[-] (vi4) to (vc3);
\draw[-] (vi5) to (vix5);
\draw[-] (vi6) to (vc4);

\draw[-] (vi7) to (vc5);
\draw[-] (vi8) to (vix8);
\draw[-] (vi9) to (vc6);
\draw[-] (vi10) to (vix10);
\draw[-] (vi11) to (vc7);
\draw[-] (vi12) to (vc8);

\draw[-] (vi13) to (vc1);
\draw[-] (vi15) to (vc2);
\draw[-] (vi17) to (vc3);
\draw[-] (vi19) to (vc4);
\draw[-] (vi14) to (vc5);
\draw[-] (vi16) to (vc6);
\draw[-] (vi18) to (vc7);
\draw[-] (vi20) to (vix20);

\draw[-] (vx4) to (vxx4);
\draw[-] (vx4) to (vix51);
\draw[-] (vx10) to (vc10);
\draw[-] (vx10) to (vxx10);
\draw[-] (vx15) to (vxx15);
\draw[-] (vx15) to (vix52);
\draw[-] (vx18) to (vc18);
\draw[-] (vx18) to (vxx18);

\draw[-] (vx191) to (vc19);
\draw[-] (vx191) to (vxx191);
\draw[-] (vx192) to (vc19);
\draw[-] (vx192) to (vxx192);

\draw[-] (vi21) to (vix21);
\draw[-] (vi22) to (vix22);
\draw[-] (vi23) to (vc8);
\draw[-] (vi24) to (vix24);

\draw[-] (vi25) to (vc9);
\draw[-] (vi26) to (vc11);
\draw[-] (vi27) to (vc10);
\draw[-] (vi28) to (vix28);

\draw[-] (vi29) to (vc11);
\draw[-] (vi30) to (vc12);
\draw[-] (vi31) to (vix31);
\draw[-] (vi32) to (vc13);
\draw[-] (vi33) to (vix33);
\draw[-] (vi34) to (vc14);

\draw[-] (vi35) to (vc15);
\draw[-] (vi36) to (vix36);
\draw[-] (vi37) to (vc16);
\draw[-] (vi38) to (vix38);
\draw[-] (vi39) to (vc17);
\draw[-] (vi40) to (vix48);

\draw[-] (vi41) to (vix41);
\draw[-] (vi42) to (vc15);
\draw[-] (vi43) to (vc12);
\draw[-] (vi44) to (vc16);
\draw[-] (vi45) to (vc13);
\draw[-] (vi46) to (vc17);
\draw[-] (vi47) to (vc14);
\draw[-] (vi48) to (vc18);

\draw[-] (vi49) to (vc9);
\draw[-] (vi50) to (vix50);

\end{tikzpicture}
\caption{\label{exemplototal}Grid graph resulting from the reduction applied to the grid graph of the Figure \ref{exemplo}.}
\end{figure}

Now, let us prove that $G$ has a path of length $\geq k$ if and only if $t(G') \geq 3k+2$.

Suppose that $G$ is a grid graph with maximum degree 3 that has a path of length $\geq k$. Let us prove that $t(G') \geq 3k+2$. Since $G$ has a path $P$ of length $\geq k$, we have that $G'$ has an induced path $P$ of length $\geq 3k$ that passes by the same path that $P$ passes, which implies that $P$ passes only by original and auxiliary vertices. Note that, when an auxiliary vertex is in $P$, his auxiliary neighbor is also in $P$.

Let $v$ and $v'$ be the extremities of $P$ and $f(v') = q'$. Since $v$ is an original vertex, then let $w$ be any auxiliary neighbor of $v$ that is not in $V(P)$. Note that all neighbors of $w$, except $v$, are not in $V(P)$. Let $r$ be the vertex auxiliary neighbor of $v'$ that is in $P$ and let $P'$ be the induced path that we obtain from $P$ by adding $w$, by removing $v'$ and by adding all vertices in any smallest path between $r$ and $q'$, excluding $r$, that only have vertices not adjacent to $w$ and passes only by original and auxiliary vertices. Since $P$ is an induced path and we removed one vertex and added only one induced path that has either 1 or 3 vertices to create $P'$, we have that $P'$ is an induced path with length $\geq 3k+1$ where all of its vertices have degree 3, except for $q'$, which has degree 2. Therefore, by Lemma \ref{lemainfcam}, we have that $t(G') \geq 3k+2$.

Now, suppose that $G$ is a grid graph with maximum degree 3 such that, when we apply the reduction to $G$ to create $G'$, we have that $t(G') \geq 3k+2$. Let us prove that $G$ has a path of length $\geq k$. Since $t(G') \geq 3k+2$, applying the Lemma \ref{lemainfcam}, we have that $G'$ has an induced path $P$ where either all vertices in $V(P)$ have degree 3 and $|E(P)| \geq 6k+2$ or all vertices in $V(P)$ have degree 3, except for one of his extremities, which has degree 2, and $|E(P)| \geq 3k+1$.

Firstly, suppose that $G'$ has an induced path $P$ where all vertices in $V(P)$ have degree 3 and $|E(P)| \geq 6k+2$. Since, the only vertices that have degree 3 are the original and auxiliary vertices and for each three consecutive vertices in $P$ there is one original vertex and two auxiliary vertices, it is easy to see that $P$ has at least $k+1$ original vertices and, thus, there is a path in $G$ of length at least $k$.

Finally, suppose that $G'$ has an induced path $P$ where all vertices in $V(P)$ have degree 3, except for one of his extremities, which has degree 2, and $|E(P)| \geq 3k+1$. It is enough to analyze the case $|E(P)| = 3k+1$ because, if $|E(P)| > 3k+1$, any subpath of $P$ of length $3k+1$ that starts at the extremity of $P$ that have degree 2 is an induced path where all of his vertices have degree 3, except for one of his extremities, which has degree 2, and has length $3k+1$. So, let us say that $P$ starts in the vertex that has degree 2. Since the only vertices that have degree 2 are corner vertices, then $P$ starts with a corner vertex. Let $q$ be that corner vertex, let $q' = f^{-1}(q)$ and let $v$ be the other extremity of $P$. 

Suppose that $P$ passes by $q'$. Since $P$ is an induced path, then $q'$ is the third vertex of $P$. Since $q$ and $q'$ are at distance 2 of each other and $|E(P)| = 3k+1$, then $v$ is an auxiliary vertex which his neighbor that is an original vertex, say $v'$, is not in $P$. Let us append $v'$ to $P$ and remove all vertices between $q$ and $q'$, including $q$ and excluding $q'$. So, since $P$ starts at $q'$, an original vertex, ends in $v'$, another original vertex, and has length $3k$, then there is a path in $G$ of length greater or equal to $k$.
%Since $P$ passes only by original and auxiliary vertices and, for each three consecutive vertices in $P$, there is one original vertex and two auxiliary vertices, then there is $k+1$ original vertices in $P$, which means that there is a path in $G$ of length greater or equal to $k$.

Now, suppose that $P$ does not pass by $q'$. Since $|E(P)| = 3k+1$, then $v$ is an auxiliary vertex which his neighbor that is an original vertex, say $v'$, is in $P$. Let us remove $q$, appending $q'$ in his place, and $v$ from $P$. Thus, since $P$ starts at $q'$, an original vertex, ends in $v'$, another original vertex, and has length $3k$, then there is a path in $G$ of length greater or equal to $k$.
%Since $P$ passes only by original and auxiliary vertices and, for each three consecutive vertices in $P$, there is one original vertex and two auxiliary vertices, then there is $k+1$ original vertices in $P$, which means that there is a path in $G$ of length greater or equal to $k$.
\end{proof}

% ----------------------------------------------------------------------------------
% ----------------------------------------------------------------------------------
% ----------------------------------------------------------------------------------
% ----------------------------------------------------------------------------------
\section{Percolation Time Problem in solid grid graphs with $\Delta = 3$}
\label{solidgriddelta3}

A solid grid graph is a grid graph in which all of his bounded faces have area one. There are NP-hard problems in grid graphs that are polynomial time solvable for solid grid graphs. For example, since 1982 it is known that the hamiltonian cycle problem is NP-hard for grid graphs \cite{itai82}, but, in 1997, it was proved that it is polynomial time solvable for solid grid graphs \cite{lenhart97}. Motivated by the work of \cite{fabricio1} on the maximum percolation time for rectangular grids, we obtain in this section a polynomial time algorithm for solid grid graphs with maximum degree 3. 
However, the Percolation Time Problem for solid grid graphs with maximum degree 4 is still open.

%In this section, we obtain a polynomial time algorithm that solves the optimization version of the Percolation Time Problem in solid grid graphs with $\Delta = 3$. A solid grid graph \cite{lenhart97} is a grid graph with no holes.

\begin{theorem}
For any solid grid graph $G$  with $\Delta = 3$, $t(G)$ can be found in $O(n^2)$ time.
\end{theorem}

\begin{proof}
If a solid grid graph has $\Delta = 3$, then, since it is $K_{1,4}$-free, it becomes a graph formed only by ladders $L_k$, which are grid graphs with dimensions $2 \times k$, and by paths, possibly linking these ladders by the vertices in their extremities. Let the extremities of a ladder be the four vertices that have only two neighbors in the ladder and let all the other vertices be the vertices internal to the ladder. In Figure \ref{solidgrid}, there is an example of solid grid graph with $\Delta = 3$.

\begin{figure}[ht]
\centering
\begin{tikzpicture}[scale=0.75]
\tikzstyle{vertex}=[draw,circle,fill=black!25,minimum size=6pt,inner sep=1pt]
\tikzstyle{intvertex}=[draw,circle,fill=black!50,minimum size=2pt,inner sep=2pt]

\node[vertex] (v1) at (2,0) {};

\foreach \i in {2,...,11}{
	\node[vertex] (v\i) at (\i-2,1) {};
}
\draw[-] (v1) to (v4);
\foreach \i [evaluate=\i as \j using (\i-1)] in {3,...,11}{
	\draw[-] (v\i) to (v\j);
}

\foreach \i in {12,...,16}{
	\node[vertex] (v\i) at (\i-8,2) {};
}
\foreach \i [evaluate=\i as \j using (\i-1)] in {13,...,16}{
	\draw[-] (v\i) to (v\j);
}
\foreach \i/\j in {6/12,7/13,8/14,9/15,10/16}{
	 \draw[-] (v\i) to (v\j);
}

\foreach \i in {17,...,19}{
	\node[vertex] (v\i) at (0,\i-15) {};
}
\draw[-] (v2) to (v17);
\foreach \i/\j in {17/18,18/19}{
	 \draw[-] (v\i) to (v\j);
}

\foreach \i in {20,...,26}{
	\node[vertex] (v\i) at (1,\i-18) {};
}
\node[vertex] (v27) at (5,6) {};
\node[vertex] (v28) at (0,6) {};

\draw[-] (v3) to (v20);
\foreach \i/\j in {20/21,21/22,22/23,23/24,24/25,25/26}{
	 \draw[-] (v\i) to (v\j);
}

\foreach \i/\j in {17/20,18/21,19/22}{
	 \draw[-] (v\i) to (v\j);
}

\draw[-] (v24) to (v28);

\foreach \i in {29,...,31}{
	\node[vertex] (v\i) at (\i-27,7) {};
}
\draw[-] (v25) to (v29);
\foreach \i/\j in {29/30,30/31}{
	 \draw[-] (v\i) to (v\j);
}

\foreach \i in {32,...,34}{
	\node[vertex] (v\i) at (\i-30,8) {};
}
\draw[-] (v26) to (v32);
\foreach \i/\j in {32/33,33/34}{
	 \draw[-] (v\i) to (v\j);
}

\foreach \i/\j in {29/32,30/33,31/34}{
	 \draw[-] (v\i) to (v\j);
}

\node[vertex] (v36) at (5,7) {};

 \draw[-] (v31) to (v36);

\foreach \i in {38,...,40}{
	\node[vertex] (v\i) at (\i-32,7) {};
}

\draw[-] (v38) to (v36);

\foreach \i in {41,...,43}{
	\node[vertex] (v\i) at (\i-35,8) {};
}

\foreach \i/\j in {38/39,39/40}{
	 \draw[-] (v\i) to (v\j);
}

\foreach \i/\j in {41/42,42/43}{
	 \draw[-] (v\i) to (v\j);
}

\foreach \i/\j in {38/41,39/42,40/43}{
	 \draw[-] (v\i) to (v\j);
}

\foreach \i in {44,...,47}{
	\node[vertex] (v\i) at (\i-35,7) {};
}
\foreach \i in {48,...,51}{
	\node[vertex] (v\i) at (\i-39,6) {};
}

\draw[-] (v40) to (v44);

\foreach \i/\j in {44/45,45/46,46/47}{
	 \draw[-] (v\i) to (v\j);
}
\foreach \i/\j in {48/49,49/50,50/51}{
	 \draw[-] (v\i) to (v\j);
}
\foreach \i/\j in {44/48,45/49,46/50,47/51}{
	 \draw[-] (v\i) to (v\j);
}

\node[vertex] (v52) at (9,5) {};
\node[vertex] (v53) at (8,5) {};

\draw[-] (v48) to (v52);
\draw[-] (v52) to (v53);

\foreach \i in {54,...,58}{
	\node[vertex] (v\i) at (\i-45,4) {};
}
\foreach \i in {59,...,61}{
	\node[vertex] (v\i) at (\i-50,3) {};
}
\node[vertex] (v62) at (11,2) {};

\draw[-] (v52) to (v54);

\foreach \i/\j in {54/55,55/56,56/57,57/58}{
	 \draw[-] (v\i) to (v\j);
}
\foreach \i/\j in {59/60,60/61}{
	 \draw[-] (v\i) to (v\j);
}
\foreach \i/\j in {54/59,55/60,56/61}{
	 \draw[-] (v\i) to (v\j);
}

\draw[-] (v61) to (v62);

\foreach \i in {65,...,67}{
	\node[vertex] (v\i) at (75-\i,9) {};
}
\foreach \i in {68,...,70}{
	\node[vertex] (v\i) at (78-\i,10) {};
}

\foreach \i/\j in {65/66,66/67}{
	 \draw[-] (v\i) to (v\j);
}
\foreach \i/\j in {68/69,69/70}{
	 \draw[-] (v\i) to (v\j);
}
\foreach \i/\j in {65/68,66/69,67/70}{
	 \draw[-] (v\i) to (v\j);
}

\draw[-] (v67) to (v43);
\draw[-] (v36) to (v27);

\node[vertex] (v71) at (11,10) {};
\node[vertex] (v72) at (5,5) {};
\node[vertex] (v73) at (6,9) {};

\draw[-] (v68) to (v71);
\draw[-] (v27) to (v72);
\draw[-] (v41) to (v73);

\end{tikzpicture}
\caption{\label{solidgrid}A solid grid graph with maximum degree 3.}
\end{figure}
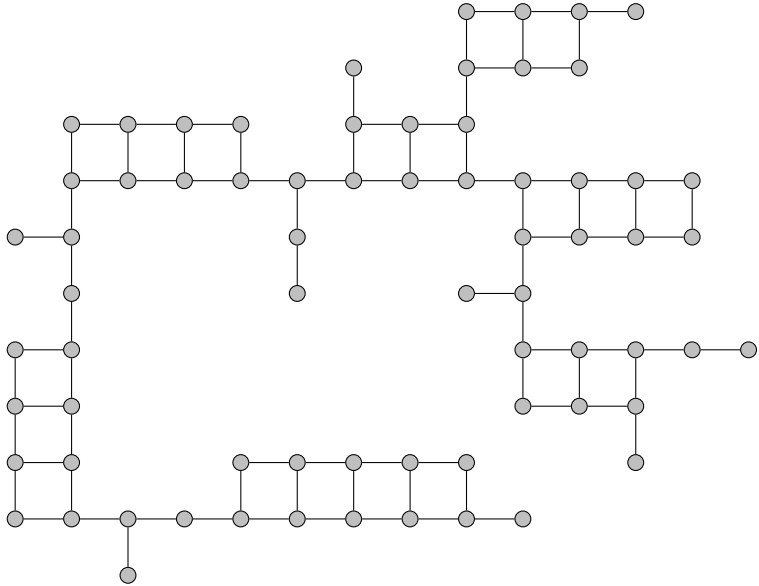

To find the percolation time of $G$, according to Lemma \ref{lemainfcam}, it is enough to find both the longest induced path that starts with a degree 2 vertex and, then, passes only by vertices with degree 3, and the longest induced path that passes only by vertices with degree 3. Thus, since all bounded faces of $G$ are squares with area one and $G$ is composed only by paths and ladders, the only difficulty to calculate $t(G)$ is to find the longest induced paths in the ladders between any two extremities that passes only by vertices with degree 3. However, one can easily calculate the longest induced paths between any two extremities of a ladder $L_k$: if the two extremities are neighbors, the length of the longest induced paths between them is 1; if the two extremities are at distance $k-1$, the length of the longest induced paths between them is $(k-t) + 2 \cdot \lfloor (k-t+1)/4 \rfloor - 1 + t$; if the two extremities are at distance $k$, the length of the longest induced paths between them is $(k-t) + 2 \cdot \lfloor (k-t-1)/4 \rfloor + t$, where $t$ is how many of the two others extremities have degree 2.

\begin{figure}[ht]
\centering
\begin{tikzpicture}[scale=0.75]
\tikzstyle{vertex}=[draw,circle,fill=black!25,minimum size=7pt,inner sep=1pt]
\tikzstyle{intvertex}=[draw,circle,fill=black!50,minimum size=2pt,inner sep=2pt]
\tikzstyle{weight} = [font=\tiny,fill=white,minimum size=1pt,inner sep=1.25pt]
\def \scale {1.15}

\node[vertex] (v1) at (0*\scale,1*\scale) {};
\node[vertex] (v2) at (1*\scale,1*\scale) {};
\node[vertex] (v3) at (0*\scale,2*\scale) {};
\node[vertex] (v4) at (1*\scale,2*\scale) {};

\draw (v1) -- node[weight] {1} (v2);
\draw (v1) -- node[weight] {5} (v3);
\draw (v1) -- node[weight] {4} (v4);
\draw (v2) -- node[weight] {4} (v3);
\draw (v2) -- node[weight] {3} (v4);
\draw (v3) -- node[weight] {1} (v4);

\node[vertex] (v5) at (2*\scale,1*\scale) {};
\node[vertex] (v6) at (2*\scale,0*\scale) {};
\node[vertex] (v7) at (3*\scale,1*\scale) {};

\draw[-] (v2) to (v5);
\draw[-] (v5) to (v6);
\draw[-] (v5) to (v7);

\node[vertex] (v8) at (4*\scale,1*\scale) {};
\node[vertex] (v9) at (5*\scale,1*\scale) {};
\node[vertex] (v10) at (4*\scale,2*\scale) {};
\node[vertex] (v11) at (5*\scale,2*\scale) {};

\draw (v8)  -- node[weight] {6} (v9);
\draw (v8)  -- node[weight] {1} (v10);
\draw (v8)  -- node[weight] {5} (v11);
\draw (v9)  -- node[weight] {5} (v10);
\draw (v9)  -- node[weight] {1} (v11);
\draw (v10) -- node[weight] {6} (v11);

\draw[-] (v7) to (v8);

\node[vertex] (v12) at (6*\scale,1*\scale) {};

\draw[-] (v9) to (v12);

\node[vertex] (v13) at (1*\scale,3*\scale) {};
\node[vertex] (v14) at (1*\scale,4*\scale) {};
\node[vertex] (v15) at (0*\scale,4*\scale) {};

\draw[-] (v4) to (v13);
\draw[-] (v13) to (v14);
\draw[-] (v14) to (v15);

\node[vertex] (v16) at (1*\scale,5*\scale) {};
\node[vertex] (v17) at (2*\scale,5*\scale) {};
\node[vertex] (v18) at (1*\scale,6*\scale) {};
\node[vertex] (v19) at (2*\scale,6*\scale) {};

\draw[-] (v14) to (v16);

\draw (v16)  -- node[weight] {3} (v17);
\draw (v16)  -- node[weight] {1} (v18);
\draw (v16)  -- node[weight] {4} (v19);
\draw (v17)  -- node[weight] {4} (v18);
\draw (v17)  -- node[weight] {1} (v19);
\draw (v18)  -- node[weight] {5} (v19);

\node[vertex] (v20) at (3*\scale,5*\scale) {};
\node[vertex] (v21) at (3*\scale,4*\scale) {};
\node[vertex] (v22) at (3*\scale,3*\scale) {};

\draw[-] (v17) to (v20);
\draw[-] (v20) to (v21);
\draw[-] (v21) to (v22);

\node[vertex] (v23) at (4*\scale,5*\scale) {};
\node[vertex] (v24) at (5*\scale,5*\scale) {};
\node[vertex] (v25) at (4*\scale,6*\scale) {};
\node[vertex] (v26) at (5*\scale,6*\scale) {};

\draw[-] (v20) to (v23);

\draw (v23)  -- node[weight] {4} (v24);
\draw (v23)  -- node[weight] {1} (v25);
\draw (v23)  -- node[weight] {3} (v26);
\draw (v24)  -- node[weight] {3} (v25);
\draw (v24)  -- node[weight] {1} (v26);
\draw (v25)  -- node[weight] {4} (v26);

\node[vertex] (v45) at (3*\scale,7*\scale) {};
\draw[-] (v25) to (v45);

\node[vertex] (v27) at (5*\scale,7*\scale) {};
\node[vertex] (v28) at (6*\scale,7*\scale) {};
\node[vertex] (v29) at (5*\scale,8*\scale) {};
\node[vertex] (v30) at (6*\scale,8*\scale) {};

\draw[-] (v26) to (v27);

\draw (v27)  -- node[weight] {2} (v28);
\draw (v27)  -- node[weight] {1} (v29);
\draw (v27)  -- node[weight] {3} (v30);
\draw (v28)  -- node[weight] {3} (v29);
\draw (v28)  -- node[weight] {1} (v30);
\draw (v29)  -- node[weight] {2} (v30);

\node[vertex] (v31) at (7*\scale,8*\scale) {};

\draw[-] (v30) to (v31);

\node[vertex] (v32) at (6*\scale,5*\scale) {};
\node[vertex] (v33) at (7*\scale,5*\scale) {};
\node[vertex] (v34) at (6*\scale,4*\scale) {};
\node[vertex] (v35) at (7*\scale,4*\scale) {};

\draw[-] (v24) to (v32);

\draw (v32)  -- node[weight] {5} (v33);
\draw (v32)  -- node[weight] {1} (v34);
\draw (v32)  -- node[weight] {4} (v35);
\draw (v33)  -- node[weight] {4} (v34);
\draw (v33)  -- node[weight] {1} (v35);
\draw (v34)  -- node[weight] {5} (v35);

\node[vertex] (v36) at (7*\scale,3*\scale) {};
\node[vertex] (v37) at (6*\scale,3*\scale) {};

\draw[-] (v34) to (v36);
\draw[-] (v36) to (v37);

\node[vertex] (v38) at (8*\scale,2*\scale) {};
\node[vertex] (v39) at (9*\scale,2*\scale) {};
\node[vertex] (v40) at (8*\scale,1*\scale) {};
\node[vertex] (v41) at (9*\scale,1*\scale) {};

\draw[-] (v36) to (v38);

\draw (v38)  -- node[weight] {2} (v39);
\draw (v38)  -- node[weight] {1} (v40);
\draw (v38)  -- node[weight] {3} (v41);
\draw (v39)  -- node[weight] {3} (v40);
\draw (v39)  -- node[weight] {1} (v41);
\draw (v40)  -- node[weight] {4} (v41);

\node[vertex] (v42) at (10*\scale,2*\scale) {};
\node[vertex] (v43) at (11*\scale,2*\scale) {};
\node[vertex] (v44) at (9*\scale,0*\scale) {};

\draw[-] (v41) to (v44);
\draw[-] (v39) to (v42);
\draw[-] (v42) to (v43);

\end{tikzpicture}
\caption{\label{solidgridtransf}The resulting graph of the transformation applied to the graph in the Figure \ref{solidgrid}.}
\end{figure}
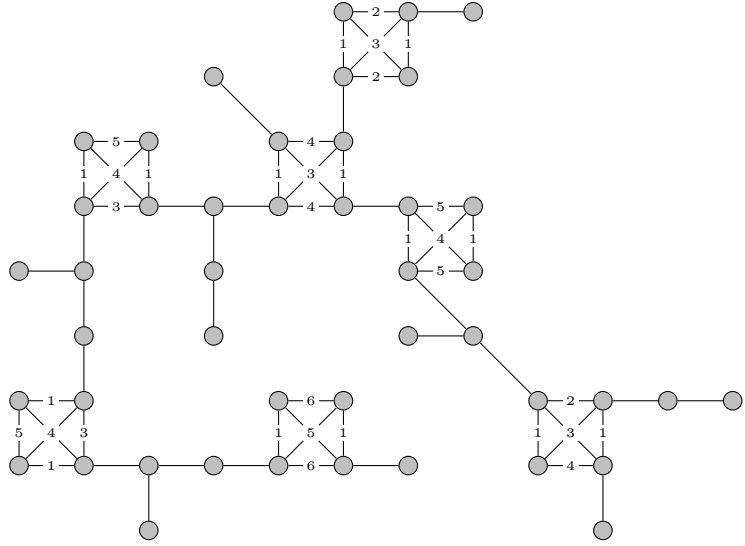

So, first, we will transform $G$ in a weighted graph $G'$ where $G'$ is the same graph as $G$ only with all the ladders replaced by weighted $K_4$'s, where the weight of an edge between two vertices in a $K_4$ represents the length of a longest induced path between the corresponding extremities of the ladders in $G$ that passes only by vertices with degree 3. The weight of all the other edges is 1. The Figure \ref{solidgridtransf} represents the transformation applied in the graph of the Figure \ref{solidgrid}. Note that there is exactly one induced path between any two vertices in $G'$, which length is equal to the longest induced path between the same two vertices in $G$. It is not hard to see that this transformation from $G$ to $G'$ can be done in linear time.

%This transformation from $G$ to $G'$ can be done in linear time. This is because, for each vertex that we do not know yet whether or not it is part of a ladder, initially, we can check whether if it is part of a square. If it is not, then it is not part of a ladder. If it is, then we can try to extend the square to the largest ladder that contains it, marking the vertices that belong to the ladder. Since $\Delta = 3$, then we can check in constant time if a given vertex is belongs to a square and, also in constant time, we can try to extend the ladder in one unit.

\begin{algorithm}[ht]
\DontPrintSemicolon
\SetKwFunction{algo}{MaximumTimeSolidGrid$\Delta{}3$}
\SetKwFunction{funo}{\texttt{LongestInducedPathFrom}}

%\SetKwProg{funLP}{Function}{}{}
%\funLP{\funo{$G',u,\text{Mark}$}}{
	%\If{Mark[u] $>$ 0}{
		%\Return {$-1$}
	%}
	%maxLength $= 0$\\
	%Mark[u] = Mark[u] + 1\\
	%$\forall w \in N(u) \bullet$ Mark[w] = Mark[w] + 1\\
	%\ForAll{$v \in N(u)$}{
		%Mark[v] = Mark[v] - 1\\
		%length $= w(u,v) +$ \funo{$G',v,\text{Mark}$}\\
		%\If{maxLength $<$ length}{
			%maxLength = length\\
		%}
		%Mark[v] = Mark[v] + 1\\
	%}
	%$\forall w \in N(u) \bullet$ Mark[w] = Mark[w] - 1\\
	 %Mark[u] = Mark[u] - 1\\
	%\Return {maxLength}
%}

\SetKwProg{myalg}{Algorithm}{}{}
\myalg{\algo{G}}{
	$G' =$ Transform$(G)$\\
	maxPercTime $= 0$\\
	\ForAll{$u \in V(G')$ such that $d_G(u) \geq 2$}{
		\If{$d_G(u) = 2$}{
			percTimeU = \funo{$G',u$}$ + 1$\\
		}
		\Else{
			percTimeU $= \lfloor($\funo{$G',u$}$ + 2)/2\rfloor$\\
		}
		\If{maxPercTime $<$ percTimeU}{
			maxPercTime $=$ percTimeU\\
		}
	}
	\Return {maxPercTime}
}

\caption{\label{algsolgriddelta3}Algorithm that finds $t(G)$ for any solid grid graph $G$ with $\Delta = 3$}
\end{algorithm}

In Algorithm \ref{algsolgriddelta3}, let $w(u,v)$ be the weight of the edge $(u,v)$. The algorithm, for each vertex $u \in V(G')$ such that $d_G(u) \geq 2$, calls the function \texttt{LongestInducedPathFrom}, which do a Depth-First Search to find the longest induced path in $G'$ from $u$ such that the last vertex is the only vertex in the path that either has degree $\leq 2$, besides perhaps the vertex $u$, or is in the neighborhood of a vertex already in the path, and, then, it subtracts the length of the found path by one. This is necessary because a longest induced path from some vertex $u$ in $G$ can end in a vertex $v$ internal to a ladder, but internal vertices of a ladder are not represented in $G'$. However, if that happens, since all vertices internal to a ladder have degree 3, then $v$ must be adjacent to some vertex at the extremity of the ladder that has degree 2.

In any case, the resulting length corresponds to the length of the longest induced path in $G$ beginning in $u$, which last vertex has degree 3 and is not in the neighborhood of any vertex already in the path. Then, it compares all these values, according to the Lemma \ref{lemainfcam}, to find $t(G)$.

Since there is only one induced path between any two vertices in $G'$, we have that the recursive function \texttt{LongestInducedPathFrom} takes the same time as any Depth-First Search algorithm. Thus, since $m = O(n)$, the Function \texttt{LongestInducedPathFrom} takes $O(n)$ time. Therefore, the Algorithm \ref{algsolgriddelta3} takes $O(n^2)$ time. 
\end{proof}

% ----------------------------------------------------------------------------------
% ----------------------------------------------------------------------------------
% ----------------------------------------------------------------------------------
% ----------------------------------------------------------------------------------
\section{Percolation Time Problem in graphs with bounded maximum degree}
\label{delta4}

In Section 2 (Theorem \ref{teo-d3-logn}), we proved that the Percolation Time Problem is polynomial time solvable in grid graphs with $\Delta(G)\leq 3$ for $k=O(\log n)$. In this section, we prove on Theorem \ref{teodelta4log} that this not happen for general graphs with fixed maximum degree $\Delta\geq 4$, unless P=NP.
However, if the percolation time $k$ is also fixed, then we prove in Theorem \ref{teo-fpt-delta} that the Percolation Time Problem is solvable in quadratic time (in other words, it is fixed parameter tractable on $\Delta(G)+k$).

%we prove that in general graphs with bounded maximum degree $\Delta\geq 4$ fixed, the Percolation Time Problem is NP-complete for any $k\geq\log_{\Delta-2}n$.

\begin{theorem}\label{teodelta4log}
Let $\Delta\geq 4$ be fixed.
Deciding whether $t(G)\geq k$ is NP-complete for graphs with bounded maximum degree $\Delta$ and any $k \geq \log_{\Delta-2}n$.
\end{theorem}

\begin{proof}

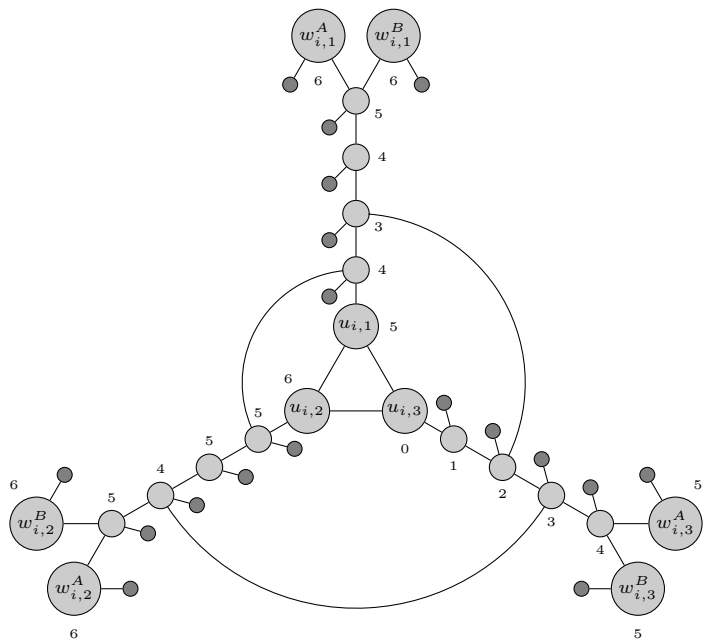
\begin{figure}[ht]
\centering
\begin{tikzpicture}[scale=0.5]
\tikzstyle{every node}=[font=\scriptsize]
\tikzstyle{vertex}=[draw,circle,fill=black!20,minimum size=10pt,inner sep=1pt]
\tikzstyle{leafvertex}=[draw,circle,fill=black!50,minimum size=2pt,inner sep=2pt]
\def \dist {1.5cm}
\foreach \i [evaluate=\i as \angle using {(\i*120 - 30)}] in {1,2,3}{
		\node[vertex] (u\i) at (\angle:\dist) {$u_{i,\i}$};
		\node[vertex] (backbone1\i) at (\angle:2*\dist) {};
		\node[leafvertex,shift={(\angle + 135:0.5cm)}] (backbone1l\i) at (backbone1\i) {};
		\node[vertex] (backbone2\i) at (\angle:3*\dist) {};
		\node[leafvertex,shift={(\angle + 135:0.5cm)}] (backbone2l\i) at (backbone2\i) {};
		\node[vertex] (backbone3\i) at (\angle:4*\dist) {};
		\node[leafvertex,shift={(\angle + 135:0.5cm)}] (backbone3l\i) at (backbone3\i) {};
		\node[vertex] (backbone4\i) at (\angle:5*\dist) {};
		\node[leafvertex,shift={(\angle + 135:0.5cm)}] (backbone4l\i) at (backbone4\i) {};
			
    \node[vertex,shift={(\angle + 30:1cm)}] (wA\i) at (backbone4\i) {$w^A_{i,\i}$};
		\node[leafvertex,shift={(\angle + 150:0.75cm)}] (wAl\i) at (wA\i) {};
		\node[vertex,shift={(\angle - 30:1cm)}] (wB\i) at (backbone4\i) {$w^B_{i,\i}$};
		\node[leafvertex,shift={(\angle -150 :0.75cm)}] (wBl\i) at (wB\i) {};
		
		\draw[-] (u\i) to (backbone1\i);
		\draw[-] (backbone1\i) to (backbone2\i);
		\draw[-] (backbone2\i) to (backbone3\i);
		\draw[-] (backbone3\i) to (backbone4\i);
		\draw[-] (backbone4\i) to (wA\i);
		\draw[-] (backbone4\i) to (wB\i);
		
		\draw[-] (backbone1\i) to (backbone1l\i);
		\draw[-] (backbone2\i) to (backbone2l\i);
		\draw[-] (backbone3\i) to (backbone3l\i);
		\draw[-] (backbone4\i) to (backbone4l\i);
		
		\draw[-] (wA\i) to (wAl\i);
		\draw[-] (wB\i) to (wBl\i);
		
	}
	
	\draw[-] (u1) to (u2);
	\draw[-] (u2) to (u3);
	\draw[-] (u3) to (u1);
	
	\draw[-] ({90 + 6.5:2*\dist}) arc (90 + 6:210 - 6.5:2*\dist);
	\draw[-] ({210 + 3.2:4*\dist}) arc (210 + 3.2:330 - 3.2:4*\dist);
	\draw[-] ({330 + 4.2:3*\dist}) arc (330 + 4.2:450 - 4.2:3*\dist);
	
	\node[shift={(270:0.5cm)}] (nu3) at (u3) {\tiny{0}};
	\node[shift={(0:0.5cm)}] (nu1) at (u1) {\tiny{5}};
	\node[shift={(120:0.5cm)}] (nu2) at (u2) {\tiny{6}};
	
	\node[shift={(270:0.6cm)}] (nwA1) at (wA1) {\tiny{6}};
	\node[shift={(270:0.6cm)}] (nwB1) at (wB1) {\tiny{6}};
	\node[shift={(270:0.6cm)}] (nwA2) at (wA2) {\tiny{6}};
	\node[shift={(120:0.6cm)}] (nwB2) at (wB2) {\tiny{6}};
	\node[shift={(60:0.6cm)}] (nwA3) at (wA3) {\tiny{5}};
	\node[shift={(270:0.6cm)}] (nwB3) at (wB3) {\tiny{5}};
	
	\node[shift={(270:0.35cm)}] (nbackbone13) at (backbone13) {\tiny{1}};
	\node[shift={(270:0.35cm)}] (nbackbone23) at (backbone23) {\tiny{2}};
	\node[shift={(270:0.35cm)}] (nbackbone33) at (backbone33) {\tiny{3}};
	\node[shift={(270:0.35cm)}] (nbackbone43) at (backbone43) {\tiny{4}};
	
	\node[shift={(0:0.35cm)}] (nbackbone11) at (backbone11) {\tiny{4}};
	\node[shift={(330:0.35cm)}] (nbackbone21) at (backbone21) {\tiny{3}};
	\node[shift={(0:0.35cm)}] (nbackbone31) at (backbone31) {\tiny{4}};
	\node[shift={(330:0.35cm)}] (nbackbone41) at (backbone41) {\tiny{5}};
	
	\node[shift={(90:0.35cm)}] (nbackbone12) at (backbone12) {\tiny{5}};
	\node[shift={(90:0.35cm)}] (nbackbone22) at (backbone22) {\tiny{5}};
	\node[shift={(90:0.35cm)}] (nbackbone32) at (backbone32) {\tiny{4}};
	\node[shift={(90:0.35cm)}] (nbackbone42) at (backbone42) {\tiny{5}};
		
\end{tikzpicture}
\caption{\label{gadget-delta3} Gadget with infection times for each clause $C_i$.}
\end{figure}

We obtain a reduction from the variation of the $\mathbf{SAT}$ problem where each clause has exactly three literals, each variable appears in at most four clauses \cite{craig1984}.

Given $M$ clauses $\mathcal{C}=\{C_1,\ldots,C_M\}$ on $N$ variables $X=\{x_1, \ldots,x_N\}$ as an instance of $\mathbf{SAT}$, we denote the three literals of $C_i$ by $\ell_{i,1}$, $\ell_{i,2}$ and $\ell_{i,3}$. Note, since any variable can only appear in at most 4 clauses, that $N/3 \leq M \leq 4N/3$. So, first, let us show how to construct a graph $G$ with maximum degree $\Delta$. For each clause $C_i$ of $\mathcal{C}$, add to $G$ a gadget as the one in Figure~\ref{gadget-delta3}. Then, for each pair of literals $\ell_{i,a}, \ell_{j,b}$ such that one is the negation of the other, add a vertex $y_{(i,a),(j,b)}$, link it to either $w^A_{i,a}$ or $w^B_{i,a}$ and link it to either $w^A_{j,b}$ or $w^B_{j,b}$, but always respecting the restriction of degree at most 4 for the vertices $w^A_{i,a}$, $w^B_{i,a}$, $w^A_{j,b}$ and $w^B_{j,b}$. Since each variable can appear in at most 4 clauses, it is always possible to do that. Let $Y$ be the set of all vertices $y_{(i,a),(j,b)}$ created this way. Notice that $y=|Y| \leq 4N$.

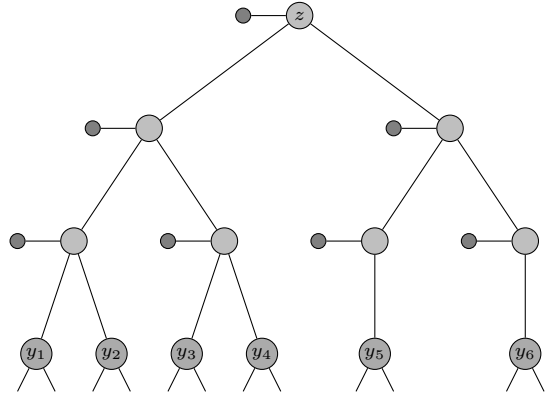
\begin{figure}[ht]
\centering
\begin{tikzpicture}[scale=0.5]
\tikzstyle{every node}=[font=\scriptsize]
\tikzstyle{vertex}=[draw,circle,fill=black!25,minimum size=10pt,inner sep=1pt]
\tikzstyle{vertexy}=[draw,circle,fill=black!33,minimum size=10pt,inner sep=1pt]
\tikzstyle{leafvertex}=[draw,circle,fill=black!50,minimum size=2pt,inner sep=2pt]

		\node[vertex] (z) at (0,9) {$z$};
		\node[leafvertex,shift={(180:0.75cm)}] (zl) at (z) {};
		\node[vertex] (z1) at (-4,6) {};
		\node[leafvertex,shift={(180:0.75cm)}] (z1l) at (z1) {};
		\node[vertex] (z2) at (4,6) {};
		\node[leafvertex,shift={(180:0.75cm)}] (z2l) at (z2) {};
		\node[vertex] (z3) at (-6,3) {};
		\node[leafvertex,shift={(180:0.75cm)}] (z3l) at (z3) {};
		\node[vertex] (z4) at (-2,3) {};
		\node[leafvertex,shift={(180:0.75cm)}] (z4l) at (z4) {};
		\node[vertex] (z5) at (2,3) {};
		\node[leafvertex,shift={(180:0.75cm)}] (z5l) at (z5) {};
		\node[vertex] (z6) at (6,3) {};
		\node[leafvertex,shift={(180:0.75cm)}] (z6l) at (z6) {};
		
		\node[vertexy] (y1) at (-7,0) {$y_1$};
		\node[vertexy] (y2) at (-5,0) {$y_2$};
		\node[vertexy] (y3) at (-3,0) {$y_3$};
		\node[vertexy] (y4) at (-1,0) {$y_4$};
		\node[vertexy] (y5) at (2,0) {$y_5$};
		\node[vertexy] (y6) at (6,0) {$y_6$};
		
		\draw[-] (z) to (z1);
		\draw[-] (z) to (z2);
		\draw[-] (z1) to (z3);
		\draw[-] (z1) to (z4);
		\draw[-] (z2) to (z5);
		\draw[-] (z2) to (z6);
		\draw[-] (z3) to (y1);
		\draw[-] (z3) to (y2);
		\draw[-] (z4) to (y3);
		\draw[-] (z4) to (y4);
		\draw[-] (z5) to (y5);
		\draw[-] (z6) to (y6);
		
		\draw[-] (z) to (zl);
		\draw[-] (z1) to (z1l);
		\draw[-] (z2) to (z2l);
		\draw[-] (z3) to (z3l);
		\draw[-] (z4) to (z4l);
		\draw[-] (z5) to (z5l);
		\draw[-] (z6) to (z6l);
		
		\draw[-] (y1) to (-7.5,-1);
		\draw[-] (y1) to (-6.5,-1);
		\draw[-] (y2) to (-5.5,-1);
		\draw[-] (y2) to (-4.5,-1);
		\draw[-] (y3) to (-3.5,-1);
		\draw[-] (y3) to (-2.5,-1);
		\draw[-] (y4) to (-1.5,-1);
		\draw[-] (y4) to (-0.5,-1);
		\draw[-] (y5) to (1.5,-1);
		\draw[-] (y5) to (2.5,-1);
		\draw[-] (y6) to (5.5,-1);
		\draw[-] (y6) to (6.5,-1);
\end{tikzpicture}
\caption{\label{linkYtoTree}Example of the full $(\Delta-2)$-ary tree added to $G$, for $\Delta = 4$ and $y = 6$.}
\end{figure}

Then, add the maximum full $(\Delta-2)$-ary tree with root $z$ such that the number of leaves is less than $y$ and, then, add a new vertex of degree one adjacent to each vertex of the tree. Let $T$ be this tree and $t=|V(T)|$. After that, link each leaf to at least one and at most $\Delta-2$ vertices in $Y$. Thus, each vertex in the tree has degree $\Delta$, except for the leaves, which have degree at most $\Delta$, and $z$, which has degree $\Delta-1$.
Figure \ref{linkYtoTree} shows an example for $\Delta=4$ and $y=6$.
Notice that $t \leq 2 \cdot \frac{(\Delta - 2)y - 1}{\Delta - 3}\leq 16N$ and $height(T) = \lceil \log_{\Delta - 2} y \rceil$.

Let $c = (\Delta-2)^{-8}$, $\alpha = xc$, where $x = \max(41 + \lceil c \rceil,\lceil 1/c \rceil)$, $r = \lceil \log_{\Delta - 2} (4N\alpha) \rceil - \lceil \log_{\Delta - 2} y \rceil$ and $\beta = 4Nx - (39M + y + t + 2 + 2r)$.
It is not difficult to see that $r$ and $\beta$ are non-negative integers.

If $r>0$, add a path with $r$ vertices, link one end to $z$ and let $q$ be the other end. Also, add a new neighbor of degree one to each vertex that belongs to the path. Let $P'$ be the set of vertices in this path and his neighbors of degree one. If $r = 0$, let $q = z$.

Finally, add a path with $\beta+2$ vertices and link one end to $q$. Let $P$ be the set of vertices in this path and let $x$ be the vertex in $P$ that is adjacent to $q$. By our construction, since $\Delta \geq 4$, we have that $G$ is a graph in which every vertex has degree at most $\Delta$.

Notice that any percolating set must contain a vertex of $\{u_{i,1},u_{i,2},u_{i,3}\}$ for each clause $C_i$ of $\mathcal{C}$ and all vertices that have degree 1.
Thus, following similar arguments presented in \cite{wg2014}, it is possible to prove that the maximum percolation time of the vertex $z$ is $height(T)+7$ if and only if $\mathcal{C}$ is satisfiable, which implies that  the maximum percolation time of the vertex $x$ is $ \lceil \log_{\Delta - 2} (4N\alpha) \rceil+8$ if and only if $\mathcal{C}$ is satisfiable. 

Then, we have that $\mathcal{C}$ is satisfiable if and only if $t(G) \geq \lceil \log_{\Delta - 2} (4N\alpha) \rceil+8$, but, since $n = |V(G)| = 39M + y + t + 2 + 2r + \beta$, then $4N\alpha = c \cdot n$. Therefore, since $c=(\Delta - 2)^{-8}$, $\mathcal{C}$ is satisfiable if and only if $t(G) \geq \lceil \log_{\Delta - 2} n \rceil$.
\end{proof}

We already saw that the Percolation Time Problem is NP-hard for graphs with maximum degree $\Delta(G)\geq 3$.
In this section, we prove that the Percolation Time Problem is FPT on $\Delta(G)+k$. In fact, we prove a stronger result:

\begin{theorem}\label{teo-fpt-delta}
Percolation Time Problem is fixed parameter tractable with parameter $\Delta(G)+k$.
Moreover, for fixed $\Delta$, the Percolation Time Problem is polynomial time solvable in graphs with bounded maximum degree $\Delta$ for $k=\log_\Delta O(\log n)$, if $\Delta\geq 4$, and for $k=O(\log n)$, if $\Delta = 3$.
\end{theorem}

\begin{proof}
Let $\Delta=\Delta(G)$ and let $u\in V(G)$. Then $|N_{\leq k}(u)|\leq\Delta^k$ and, consequently, the power set $2^{N_{\leq k}(u)}$ has $2^{|N_{\leq k}(u)|}\leq 2^{\Delta^k}$ sets.
We claim that $t(G)\geq k$ if and only if there is a vertex $u$ and a percolating set $S \supseteq N_{\geq k}(u)$ such that $t(G,S,u)=k$.

If $t(G)\geq k$, then there is a percolating set $S'$ that infects some vertex $u$ at time $k$.
In \cite{wg2014}, it was proved that, given a graph $G$, a set $Q\subseteq V(G)$ and a vertex $z\in V(G)\setminus S$, if $t(G,Q,w)\geq k$, then $t(G,Q,w)\geq t(G,Q\cup\{z\},w)\geq k$, for any $k$ and any $w\in N_{\geq k}(z)$.
Then, applying this result once for each vertex in $N_{\geq k}(u)$, the percolating set $S = S'\cup N_{\geq k}(u)$ infects $u$ also at time $k$.

On the other hand, if there is a percolating set $S \supseteq N_{\geq k}(u)$ such that $t(G,S,u)=k$, for some vertex $u$, then, trivially, $t(G)\geq k$. Then the claim is true.

Therefore, since for each vertex $u$ and set $S' \subseteq N_{\leq k-1}(u)$, it takes $O(km)$ time to know whether the set $S' \cup N_{\geq k}(u)$ infects $u$ at time $k$, this equivalence gives us an algorithm that decides whether $t(G)\geq k$ in time $n\cdot O(m + km\cdot 2^{\Delta^k}) = O(2^{\Delta^k}k\Delta \cdot n^2)$, since $m=O(\Delta n)$. Notice that, if $k=\log_\Delta O(\log n)$, then the time is polynomial in $n$. Moreover, if $\Delta = 3$, by Theorem \ref{teo-d3-logn}, we are done.
\end{proof}

% ----------------------------------------------------------------------------------
% ----------------------------------------------------------------------------------
% ----------------------------------------------------------------------------------
% ----------------------------------------------------------------------------------
\section{Fixed Parameter Tractability on the treewidth}
\label{fpt}

We say that a decision problem is \emph{fixed parameter tractable} (or just \emph{fpt}) on some parameter $\Psi$ if there exists an algorithm (called \emph{fpt-algorithm}) that solves the problem in time $f(\Psi)\cdot n^{O(1)}$, where $n$ is the size of the input and $f$ is an arbitrary function depending only on the parameter $\Psi$.

In this section, we obtain some results regarding the fixed parameter tractability of the Percolation Time Problem when the parameter is the treewidth $tw(G)$ of the graph. We prove on Lemma \ref{fpt-tw1} that the Percolation Time Problem is FPT on $tw(G)+k$ (that is, to decide if the time is at most $k$). In fact, we prove a stronger (but technical) result on Theorem \ref{fpt-tw3}: the Percolation Time Problem is polynomial time solvable for fixed $tw(G)$ and it is linear time solvable for fixed $tw(G)+k$. And what happens when $tw(G)$ is fixed, but $k$ is not fixed? The problem is FPT on the treewidth? We prove in Theorem \ref{fpt-tw2} that the Percolation Time Problem is W[1]-hard when parameterized by the treewidth.

\begin{lemma}\label{fpt-tw1}
The Percolation Time Problem is fixed parameter tractable with parameter $cwd(G)+k$, where $cwd(G)$ is the clique-width of $G$. Consequently, the Percolation Time Problem is fpt on $tw(G)+k$.
\end{lemma}

\begin{proof}
A consequence of the Courcelle's theorem \cite{fpt1,fpt2} states that, if a decision problem on graphs can be expressed in a Monadic Second Order (MSO$_1$) sentence $\varphi$ (with quantification only over vertex subsets), then this problem is fixed parameter tractable in the parameter $cwd(G)+|\varphi|$, where $cwd(G)$ is the clique-width of $G$. It is known that fixed parameter tractability on the clique-width implies fixed parameter tractability on the treewidth. Moreover, the running time is linear on the size of the input. The Percolation Time Problem can be expressed by the following MSO$_1$-sentence:
\[
   maxtime_k\ :=\ \exists w,X_0,X_1,\ldots,X_k\ \forall x\ \Big(x\in X_k\Big)\ \wedge\ \left(\bigwedge_{0\leq i<k}(x\in X_i\to x\in X_{i+1})\right)\ \wedge
\]
\[
\wedge\left(\bigwedge_{0\leq i<k}(x\in X_{i+1}\setminus X_i)\to\exists y,z(Exy\wedge Exz\wedge (y\in X_i)\wedge (z\in X_i))\right)\wedge\ \Big(w\in X_k\setminus X_{k-1}\Big),
\]
where $X_i$ represents the set of vertices infected at time $i$, $Exy$ is true if $xy$ is an edge (and false, otherwise) and $\wedge$ is the \emph{and} operator.
This MSO$_1$ sentence asserts that all vertices are infected in time $k$, that a vertex infected in time $i$ remains infected in time $i+1$, that a vertex infected in time $i+1$, but not infected in time $i$, has two neighbors infected in time $i$, and that there exists a vertex $w$ infected in time $k$ but not infected in time $k-1$.
\end{proof}

The proof of the next theorem uses some technicalities inspired by \cite{Fellows11}.

\begin{theorem}\label{fpt-tw2}
The Percolation Time Problem is W[1]-hard with parameter $tw(G)$.
\end{theorem}

\begin{proof}
Let us prove that the percolation time problem is $W[1]$-hard when parametrized by $treewidth$. We found a reduction from the Multicolored Clique Problem, which is $W[1]$-hard when parametrized by the number of colors $k$ \cite{multihard}. Given a graph $G$, an integer $k$ and a partition $V_1,V_2,\hdots,V_k$ of $V(G)$, which represent the colors of the vertices, the Multicolored Clique Problem asks for a $k$-clique containing exactly one vertex from each partition $V_i$.

Let us show a parameterized reduction from an instance $(G,k,C)$ of Multicolored Clique Problem to an instance $(G')$ of Percolation Time Problem, where the treewidth of $G'$ is bounded by $k' = 2k+3$.

To save some space in the figures, every time we have a path where every vertex of this path has one pendant neighbor, we will replace it by a weighted edge, like in Figure \ref{path}. So, if the vertex $i$ is infected at time $t$ by some set, then, if the infection follows the path $v_1,v_2,\hdots,v_{n-1}$, the vertex $v_{n-1}$ will be infected at time $t + (n-1)$ by the same set.

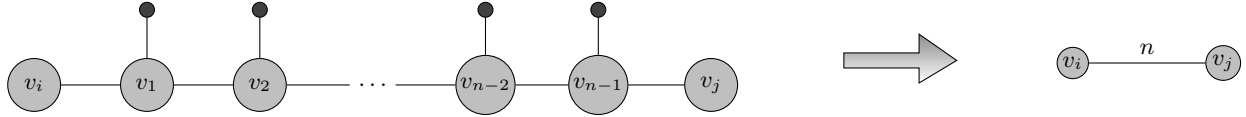
\begin{figure}[ht]

\begin{minipage}[c]{.60\textwidth}
\centering
\begin{tikzpicture}[scale=1]
\tikzstyle{vertex}=[draw,circle,fill=black!25,minimum size=20pt,inner sep=1pt]
\tikzstyle{pendant}=[draw,circle,fill=black!75,minimum size=2pt,inner sep=2pt]

\node[vertex] (vi)  at (-4,0) {$v_i$};
\node[vertex] (v1)  at (-2.5,0) {$v_1$};
\node[pendant,shift={(90:1)}] (v1p) at (v1) {};
\node[vertex] (v2)  at (-1,0) {$v_2$};
\node[pendant,shift={(90:1)}] (v2p) at (v2) {};
\node (dots) at (0.5,0) {$\hdots$};
\node[vertex] (vn2) at (2,0) {$v_{n-2}$};
\node[pendant,shift={(90:1)}] (vn2p) at (vn2) {};
\node[vertex] (vn1) at (3.5,0) {$v_{n-1}$};
\node[pendant,shift={(90:1)}] (vn1p) at (vn1) {};
\node[vertex] (vj)  at (5,0) {$v_j$};

\draw (vi) -- (v1) -- (v2) -- (dots) -- (vn2) -- (vn1) -- (vj);
\draw (v1) -- (v1p);
\draw (v2) -- (v2p);
\draw (vn2) -- (vn2p);
\draw (vn1) -- (vn1p);
\end{tikzpicture}
\end{minipage}%
\begin{minipage}[c]{.25\textwidth}
\centering
\begin{tikzpicture}[scale=1]
\def\arrow{
(-0.5,0.1) -- (0.5,0.1) -- (0.5,0.35) -- (1,0) -- (0.5,-0.3) -- (0.5,-0.1) -- (-0.5,-0.1) -- cycle
}
\draw[color=black, top color=gray] \arrow;

\end{tikzpicture}
\end{minipage}%
\begin{minipage}[c]{.15\textwidth}
\centering
\begin{tikzpicture}[scale=1,auto]
\tikzstyle{vertex}=[draw,circle,fill=black!25,minimum size=10pt,inner sep=1pt]

\node[vertex] (vi) at (-1,0) {$v_i$};
\node[vertex] (vj) at (1,0) {$v_j$};

\draw[-] (vi) edge node {$n$} (vj);

\end{tikzpicture}
\end{minipage}%
\caption{\label{path}A weighted edge on the right and the real structure on the left}
\end{figure}

Let us describe how to construct the graph $G'$. Let $n = |V(G)|$. First, assign the labels $v_1,v_2,\hdots, v_n$ arbitrarily to the vertices in $G$ and, for $1 \leq i \leq n$, let $M_i^a = 2n + 2i - 2$ and $M_i^b = 2n - 2i + 2$. For $1 \leq x \leq k$ we are going to add the gadget $Q^x$ in the Figure \ref{ColorGadget}, where the vertices $v_i, v_j, \hdots$ are the vertices in the partition $V_x$ in the graph $G$. Let $T = \bigcup_{1\leq i \leq k} \{a^i,b^i\}$. Also, let $V^x$ be the set of all vertices in $\{v_1,v_2,\hdots,v_n\} \cap V(Q^x)$.

\begin{figure}[ht]
\centering
\begin{tikzpicture}[scale=1,auto]

\def \scale {1}

\tikzstyle{vertex}=[draw,circle,fill=black!25,minimum size=10pt,inner sep=1pt]
\tikzstyle{pendant}=[draw,circle,fill=black!75,minimum size=5pt,inner sep=1pt]

\node[vertex] (v1) at (-3,2.5) {$v_i$};
\node[pendant] (v1o) at (-2.25,2.5) {};
\node[pendant,shift={(90:0.5)}] (v1op) at (v1o) {};
\node[vertex] (v2) at (-3,1) {$v_j$};
\node[pendant] (v2o) at (-2.25,1) {};
\node[pendant,shift={(90:0.5)}] (v2op) at (v2o) {};
\node (dots) at (-3,0) {$\vdots$};
\node[vertex] (v3) at (-3,-1) {$v_h$};
\node[pendant] (v3o) at (-2.25,-1) {};
\node[pendant,shift={(90:0.5)}] (v3op) at (v3o) {};

\node[vertex] (a) at (1.5,2.5) {$a^x$};
\node[shift={(170:0.4)}] (dots4) at (a) {\vdots};
\node[pendant,shift={(90:0.5)}] (ap) at (a) {};
\node[vertex] (b) at (1.5,-0.25) {$b^x$};
\node[shift={(170:0.4)}] (dots5) at (b) {\vdots};
\node[pendant,shift={(270:0.5)}] (bp) at (b) {};

\draw[-] (v1) to (v2);
\draw[-] (v2) to (-3,0.2);
\draw[-] (dots) to (v3);

\draw[-] (v1) to (v1o);
\draw[-] (v1o) to (v1op);
\draw[-] (v1o) edge [bend left=35] node {$M_i^a-1$} (a);
\draw[-] (v1o) edge node {$M_i^b-1$} (b);

\draw[-] (v2) to (v2o);
\draw[-] (v2o) to (v2op);

\draw[-] (v3) to (v3o);
\draw[-] (v3o) to (v3op);

\draw[-] (v2o) to (-1.5,1.3);
\draw[-] (v2o) to (-1.5,0.7);
\draw[-] (v3o) to (-1.5,-0.4);
\draw[-] (v3o) to (-1.5,-0.8);

\draw[-] (a) to (1,2.1);
\draw[-] (a) to (ap);
\draw[-] (b) to (1,-0.6);
\draw[-] (b) to (bp);

\end{tikzpicture}
\caption{\label{ColorGadget}Gadget $Q^x$ we put $k$ times where each $v_i$ represents the vertex $v_i \in V_x$ of the graph $G$.}
\end{figure}
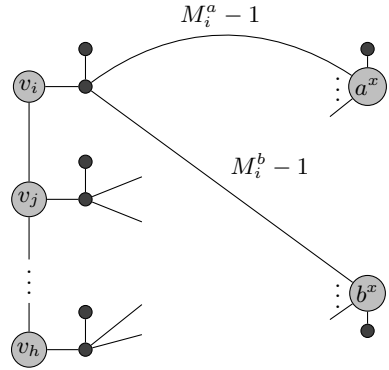

For each pair of vertices $v_i \in V_x$ and $v_j \in V_y$, for $x \neq y$, such that $(v_i,v_j) \notin E(G)$, we add a gadget such as the one in the Figure \ref{ChoiceGadget}. Also, we add six vertices $a_i^j,b_i^j,a_j^i,b_j^i$ and link $a_i^j$ to $a^x$, $b_i^j$ to $b^x$, $a_j^i$ to $a^y$ and $b_j^i$ to $b^y$. After that, link the vertex $vo_i^j$ to $a_i^j$ and to $b_i^j$ with edges with weights $M_i^a$ and $M_i^b$, respectively, and link the vertex $vo_j^i$ to $a_j^i$ and to $b_j^i$ with edges with weights $M_j^a$ and $M_j^b$, respectively. Let the set $VC_i^j = VC_j^i = \{vc_i^j,vc_j^i\}$.

In the Figure \ref{exemploChoiceGadget}, there is an example where $v_i,v_s \in V(Q^x)$ and $(v_i,v_j),(v_s,v_r) \notin E(G)$.

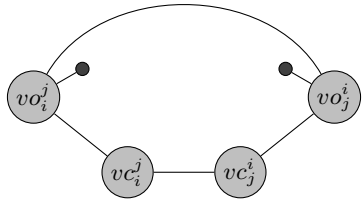
\begin{figure}[ht]
\centering
\begin{tikzpicture}[scale=1]

\tikzstyle{vertex}=[draw,circle,fill=black!25,minimum size=15pt,inner sep=1pt]
\tikzstyle{pendant}=[draw,circle,fill=black!75,minimum size=5pt,inner sep=1pt]

\node[vertex] (vc1) at (-0.75,-1) {$vc_i^j$};
\node[vertex] (vc2) at (0.75,-1) {$vc_j^i$};
\node[vertex] (vco1) at (-2,0) {$vo_i^j$};
\node[vertex] (vco2) at (2,0) {$vo_j^i$};
\node[pendant,shift={(30:0.75)}] (vco1p) at (vco1) {};
\node[pendant,shift={(150:0.75)}] (vco2p) at (vco2) {};

\draw[-] (vc1) to (vc2);
\draw[-] (vc1) to (vco1);
\draw[-] (vc2) to (vco2);
\draw[-] (vco1) to (vco1p);
\draw[-] (vco2) to (vco2p);
\draw[-] (vco1) to [bend left=60] (vco2);

\end{tikzpicture}
\caption{\label{ChoiceGadget}Gadget we put between the Graph Gadgets whenever there is no edge between the vertices $v_i$ and $v_j$.}
\end{figure}

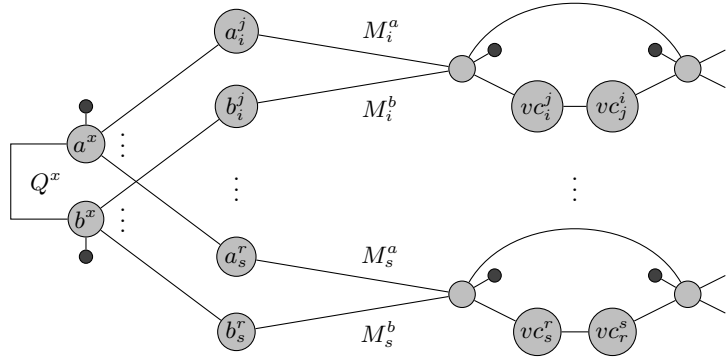
\begin{figure}[ht]
\centering
\begin{tikzpicture}[scale=1,auto]

\def \scale {1}

\tikzstyle{vertex}=[draw,circle,fill=black!25,minimum size=10pt,inner sep=1pt]
\tikzstyle{pendant}=[draw,circle,fill=black!75,minimum size=5pt,inner sep=1pt]

\node[vertex] (a) at (-3,0.5) {$a^x$};
\node[pendant,shift={(90:0.5)}] (ap) at (a) {};
\node[shift={(10:0.5)}] () at (a) {$\vdots$};
\node[vertex] (b) at (-3,-0.5) {$b^x$};
\node[pendant,shift={(270:0.5)}] (bp) at (b) {};
\node[shift={(10:0.5)}] () at (b) {$\vdots$};

\node[vertex] (aij) at (-1,2) {$a_i^j$};
\node[vertex] (bij) at (-1,1) {$b_i^j$};
\node () at (-1,0) {$\vdots$};
\node[vertex] (asr) at (-1,-1) {$a_s^r$};
\node[vertex] (bsr) at (-1,-2) {$b_s^r$};

\node[vertex] (vcix) at (3,1) {$vc_i^j$};
\node[vertex] (vcjy) at (4,1) {$vc_j^i$};
\node[vertex] (vcixo) at (2,1.5) {};
\node[vertex] (vcjyo) at (5,1.5) {};
\node[pendant,shift={(30:0.5)}] (vcixop) at (vcixo) {};
\node[pendant,shift={(150:0.5)}] (vcjyop) at (vcjyo) {};

\node () at (3.5,0) {$\vdots$};

\node[vertex] (vcsx) at (3,-2) {$vc_s^r$};
\node[vertex] (vcry) at (4,-2) {$vc_r^s$};
\node[vertex] (vcsxo) at (2,-1.5) {};
\node[vertex] (vcryo) at (5,-1.5) {};
\node[pendant,shift={(30:0.5)}] (vcsxop) at (vcsxo) {};
\node[pendant,shift={(150:0.5)}] (vcryop) at (vcryo) {};

\node[shift={(225:0.75)}] () at (a) {$Q^x$};

%------------------------------------

\draw[-] (a) to (aij);
\draw[-] (a) to (asr);
\draw[-] (a) to (ap);
\draw[-] (b) to (bij);
\draw[-] (b) to (bsr);
\draw[-] (b) to (bp);

\draw[-] (vcixo) edge node [swap] {$M_i^a$} (aij);
\draw[-] (vcixo) edge node {$M_i^b$} (bij);

\draw[-] (vcsxo) edge node [swap] {$M_s^a$} (asr);
\draw[-] (vcsxo) edge node {$M_s^b$} (bsr);

\draw[-] (vcix) to (vcjy);
\draw[-] (vcix) to (vcixo);
\draw[-] (vcjy) to (vcjyo);
\draw[-] (vcixo) to (vcixop);
\draw[-] (vcjyo) to (vcjyop);
\draw[-] (vcixo) to [bend left=60] (vcjyo);

\draw[-] (vcsx) to (vcry);
\draw[-] (vcsx) to (vcsxo);
\draw[-] (vcry) to (vcryo);
\draw[-] (vcsxo) to (vcsxop);
\draw[-] (vcryo) to (vcryop);
\draw[-] (vcsxo) to [bend left=60] (vcryo);

\draw[-] (vcjyo) to (5.5,1.75);
\draw[-] (vcjyo) to (5.5,1.25);

\draw[-] (vcryo) to (5.5,-1.75);
\draw[-] (vcryo) to (5.5,-1.25);

\draw (a) -- (-4,0.5) -- (-4,-0.5) -- (b);

\end{tikzpicture}
\caption{\label{exemploChoiceGadget}Example of Choice Gadgets that link the Gadget $Q^x$ to other gadgets.}
\end{figure}

Let $E$ be how many vertices we added until this point, including the vertices hidden by the weighted edges. We continue the construction of $G'$ by adding, for each vertex $a_i^j$ and $b_i^j$, a vertex $l_i^j$ and link both $a_i^j$ and $b_i^j$ to it with edges with weights $E-M_i^a-2$ and $E-M_i^b-2$, respectively. Let $L$ be the set of all vertices $l_i^j$ added in that manner. Then, add a vertex $p$, add another vertex and link it to $p$ and also link all vertices in $L$ to $p$. Let $D$ be how many vertices we added until this point. Finally, add a vertex $z$, add another vertex and link it to $z$ and link $p$ to $z$ with an edge with weight $D$.

In the Figure \ref{linksFinale} there is an example for this step of the construction. Since $V(G') = O(n^5)$, then $G'$ can be constructed in polynomial time.

Note that is necessary and sufficient to choose one vertex on each set $V^x$ and each set $VC_i^j$ to be part of a set $S$ in order to $S$ to infect the whole graph $G'$.

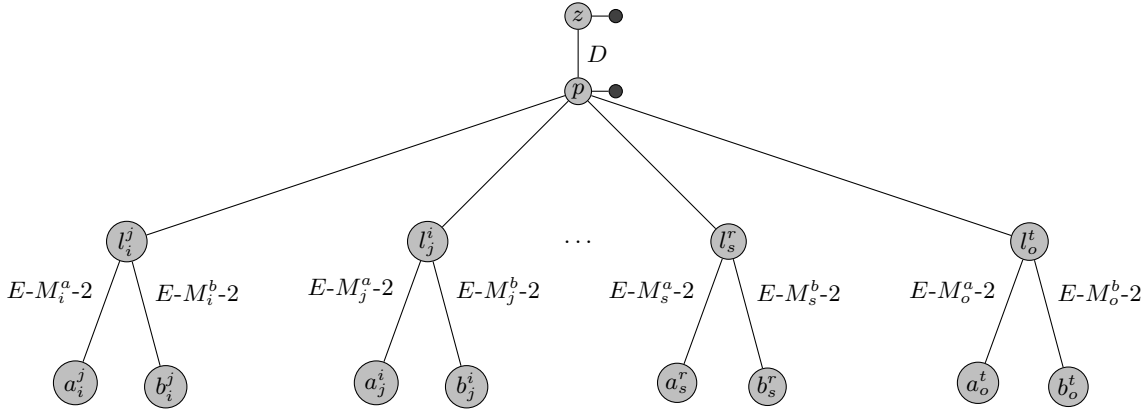
\begin{figure}[ht]
\centering
\begin{tikzpicture}[scale=1,auto]

\def \scale {1}

\tikzstyle{vertex}=[draw,circle,fill=black!25,minimum size=10pt,inner sep=1pt]
\tikzstyle{pendant}=[draw,circle,fill=black!75,minimum size=5pt,inner sep=1pt]

\node[vertex] (lijx) at (-6,-1) {$l_i^j$};
\node[vertex] (lhjx) at (-2,-1) {$l_j^i$};
\node (dots) at (0,-1) {$\hdots$};
\node[vertex] (lsry) at (2,-1) {$l_s^r$};
\node[vertex] (ltto) at (6,-1) {$l_o^t$};
\node[vertex] (p) at (0,1) {$p$};
\node[pendant,shift={(0:0.5)}] (pp) at (p) {};
\node[vertex] (z) at (0,2) {$z$};
\node[pendant,shift={(0:0.5)}] (zp) at (z) {};

\node[vertex,shift={(250:2)}] (aijx) at (lijx) {$a_i^j$};
\node[vertex,shift={(285:2)}] (bijx) at (lijx) {$b_i^j$};

\node[vertex,shift={(250:2)}] (ahjx) at (lhjx) {$a_j^i$};
\node[vertex,shift={(285:2)}] (bhjx) at (lhjx) {$b_j^i$};

\node[vertex,shift={(250:2)}] (asry) at (lsry) {$a_s^r$};
\node[vertex,shift={(285:2)}] (bsry) at (lsry) {$b_s^r$};

\node[vertex,shift={(250:2)}] (atto) at (ltto) {$a_o^t$};
\node[vertex,shift={(285:2)}] (btto) at (ltto) {$b_o^t$};

\draw[-] (lijx) edge node [swap] {$E$-$M_i^a$-2} (aijx);
\draw[-] (lijx) edge node {$E$-$M_i^b$-2} (bijx);

\draw[-] (lhjx) edge node [swap] {$E$-$M_j^a$-2} (ahjx);
\draw[-] (lhjx) edge node {$E$-$M_j^b$-2} (bhjx);

\draw[-] (lsry) edge node [swap] {$E$-$M_s^a$-2} (asry);
\draw[-] (lsry) edge node {$E$-$M_s^b$-2} (bsry);

\draw[-] (ltto) edge node [swap] {$E$-$M_o^a$-2} (atto);
\draw[-] (ltto) edge node {$E$-$M_o^b$-2} (btto);

\draw[-] (lijx) to (p);
\draw[-] (lhjx) to (p);
\draw[-] (lsry) to (p);
\draw[-] (ltto) to (p);

\draw[-] (p) to (pp);
\draw[-] (z) to (zp);

\draw[-] (p) edge node [swap] {$D$} (z);

\end{tikzpicture}
\caption{\label{linksFinale}Gadget to link all vertices in $L$}
\end{figure}

\begin{figure}[ht]
\centering
\begin{tikzpicture}[scale=1,auto]

\tikzstyle{vertex}=[draw,circle,fill=black!25,minimum size=10pt,inner sep=1pt]

\node[vertex] (v1) at (90:1.5) {$v_1$};
\node[shift={(270:0.5)}] (v1c) at (v1) {1};
\node[vertex] (v2) at (162:1.5) {$v_2$};
\node[shift={(180:0.5)}] (v2c) at (v2) {3};
\node[vertex] (v3) at (234:1.5) {$v_3$};
\node[shift={(180:0.5)}] (v3c) at (v3) {3};
\node[vertex] (v4) at (306:1.5) {$v_4$};
\node[shift={(0:0.5)}] (v4c) at (v4) {2};
\node[vertex] (v5) at (18:1.5) {$v_5$};
\node[shift={(0:0.5)}] (v5c) at (v5) {2};

\draw (v5) -- (v1) -- (v2) -- (v3) -- (v4) -- (v5) -- (v2);

\end{tikzpicture}
\caption{\label{GraphG}Graph $G$}
\end{figure}

In Figure \ref{GraphG} there is an example of Instance for the Multicolor Clique Problem for 3 partitions. The Figures \ref{GraphGl1} and \ref{GraphGl2} are the graph $G'$, which is the result of our construction applied to the graph in the Figure \ref{GraphG}.

\begin{figure}[ht]

\begin{tikzpicture}[scale=0.98,auto]

\def \vua {-23}
\def \vda {-8}
\def \vta {7}
\def \vqa {22}

\def \dab {5}

\tikzstyle{vertex}=[draw,circle,fill=black!25,minimum size=10pt,inner sep=1pt]
\tikzstyle{pendant}=[draw,circle,fill=black!75,minimum size=5pt,inner sep=1pt]

\node[vertex] (v1) at (90:9) {$v_1$};
\node[pendant] (v1o) at (90:8.5) {};
\node[pendant,shift={(30:0.5)}] (v1op) at (v1o) {};

\node[vertex] (a1) at (90-\dab-2:7.25) {$a^1$};
\node[pendant,shift={(45:0.5)}] (a1p) at (a1) {};
\node[vertex] (b1) at (90+\dab+2:7.25) {$b^1$};
\node[pendant,shift={(135:0.5)}] (b1p) at (b1) {};

\node[vertex] (v4) at (190:8.75) {$v_4$};
\node[pendant,shift={(55:0.5)}] (v4o) at (v4) {};
\node[pendant,shift={(150:0.5)}] (v4op) at (v4o) {};
\node[vertex] (v5) at (205:7.6) {$v_5$};
\node[pendant,shift={(55:0.5)}] (v5o) at (v5) {};
\node[pendant,shift={(305:0.5)}] (v5op) at (v5o) {};

\node[vertex] (a2) at (190-\dab:7) {$a^2$};
\node[pendant,shift={(180:0.5)}] (a2p) at (a2) {};
\node[vertex] (b2) at (190+\dab:6.25) {$b^2$};
\node[pendant,shift={(270:0.5)}] (b2p) at (b2) {};

\node[vertex] (v2) at (350:8.75) {$v_2$};
\node[pendant,shift={(125:0.5)}] (v2o) at (v2) {};
\node[pendant,shift={(25:0.5)}] (v2op) at (v2o) {};
\node[vertex] (v3) at (334:7.6) {$v_3$};
\node[pendant,shift={(125:0.5)}] (v3o) at (v3) {};
\node[pendant,shift={(235:0.5)}] (v3op) at (v3o) {};

\node[vertex] (b3) at (350-\dab:6.25) {$b^3$};
\node[pendant,shift={(270:0.5)}] (b3p) at (b3) {};
\node[vertex] (a3) at (350+\dab:7) {$a^3$};
\node[pendant,shift={(0:0.5)}] (a3p) at (a3) {};
%-------------------------1-2------------------------------------

\node[vertex] (vc14) at (140:6) {$vc_1^4$};
\node[vertex] (vc41) at (150:6) {$vc_4^1$};
\node[vertex] (vco14) at (135:7) {};
\node[vertex] (vco41) at (155:7) {};
\node[pendant,shift={(260:0.5)}] (vco14p) at (vco14) {};
\node[pendant,shift={(30:0.5)}] (vco41p) at (vco41) {};

\node[vertex] (a14) at (120:7) {$a_1^4$};
\node[vertex] (b14) at (120:6) {$b_1^4$};
\node[vertex] (a41) at (170:7) {$a_4^1$};
\node[vertex] (b41) at (170:6) {$b_4^1$};

%------------------------1-3-------------------------------------

\node[vertex] (vc13) at (40:6) {$vc_1^3$};
\node[vertex] (vc31) at (30:6) {$vc_3^1$};
\node[vertex] (vco13) at (45:7) {};
\node[vertex] (vco31) at (25:7) {};
\node[pendant,shift={(280:0.5)}] (vco13p) at (vco13) {};
\node[pendant,shift={(150:0.5)}] (vco31p) at (vco31) {};

\node[vertex] (a13) at (60:7) {$a_1^3$};
\node[vertex] (b13) at (60:6) {$b_1^3$};
\node[vertex] (a31) at (10:7) {$a_3^1$};
\node[vertex] (b31) at (10:6) {$b_3^1$};

%------------------------2-3-------------------------------------

\node[vertex] (vc42) at (-0.75,-0.5) {$vc_4^2$};
\node[vertex] (vc24) at (0.75,-0.5) {$vc_2^4$};
\node[vertex] (vco42) at (-1.25,0.5) {};
\node[vertex] (vco24) at (1.25,0.5) {};
\node[pendant,shift={(325:0.5)}] (vco42p) at (vco42) {};
\node[pendant,shift={(215:0.5)}] (vco24p) at (vco24) {};

\node[vertex] (a42) at (-3.5,0.5) {$a_4^2$};
\node[vertex] (b42) at (-3.5,-0.5) {$b_4^2$};
\node[vertex] (a24) at (3.5,0.5) {$a_2^4$};
\node[vertex] (b24) at (3.5,-0.5) {$b_2^4$};

\node[vertex] (vc53) at (-0.75,-3) {$vc_5^3$};
\node[vertex] (vc35) at (0.75,-3) {$vc_3^5$};
\node[vertex] (vco53) at (-1.25,-2) {};
\node[vertex] (vco35) at (1.25,-2) {};
\node[pendant,shift={(325:0.5)}] (vco53p) at (vco53) {};
\node[pendant,shift={(215:0.5)}] (vco35p) at (vco35) {};

\node[vertex] (a53) at (-3.5,-2) {$a_5^3$};
\node[vertex] (b53) at (-3.5,-3) {$b_5^3$};
\node[vertex] (a35) at (3.5,-2) {$a_3^5$};
\node[vertex] (b35) at (3.5,-3) {$b_3^5$};

%----------------------------------------------------------------

\draw[-] (v1) to (v1o);
\draw[-] (v1o) to (v1op);
\draw[-] (v1o) edge node {$M_1^a$-1} (a1);
\draw[-] (v1o) edge node [swap]{$M_1^b$-1} (b1);
\draw[-] (a1) to (a1p);
\draw[-] (b1) to (b1p);

\draw[-] (v4) to (v4o);
\draw[-] (v4o) to (v4op);
\draw[-] (v4o) edge node {$M_4^a$-1} (a2);
\draw[-] (v4o) edge node [swap]{$M_4^b$-1} (b2);
\draw[-] (v5) to (v5o);
\draw[-] (v5o) to (v5op);
\draw[-] (v5o) edge node {$M_5^a$-1} (a2);
\draw[-] (v5o) edge node [swap]{$M_5^b$-1} (b2);
\draw[-] (a2) to (a2p);
\draw[-] (b2) to (b2p);
\draw[-] (v4) to (v5);

\draw[-] (v2) to (v2o);
\draw[-] (v2o) to (v2op);
\draw[-] (v2o) edge node [swap]{$M_2^a$-1} (a3);
\draw[-] (v2o) edge node {$M_2^b$-1} (b3);
\draw[-] (v3) to (v3o);
\draw[-] (v3o) to (v3op);
\draw[-] (v3o) edge node [swap]{$M_3^a$-1} (a3);
\draw[-] (v3o) edge node {$M_3^b$-1} (b3);
\draw[-] (a3) to (a3p);
\draw[-] (b3) to (b3p);
\draw[-] (v2) to (v3);

\draw[-] (a1) to (a13);
\draw[-] (a1) to (a14);
\draw[-] (b1) to (b13);
\draw[-] (b1) to (b14);

\draw[-] (a2) to (a41);
\draw[-] (a2) to (a42);
\draw[-] (a2) to (a53);
\draw[-] (b2) to (b41);
\draw[-] (b2) to (b42);
\draw[-] (b2) to (b53);

\draw[-] (a3) to (a31);
\draw[-] (a3) to (a24);
\draw[-] (a3) to (a35);
\draw[-] (b3) to (b31);
\draw[-] (b3) to (b24);
\draw[-] (b3) to (b35);

\draw[-] (vc14) to (vc41);
\draw[-] (vc14) to (vco14);
\draw[-] (vc41) to (vco41);
\draw[-] (vco14) to (vco14p);
\draw[-] (vco41) to (vco41p);
\draw[-] (vco14) to (vco41);

\draw[-] (vc13) to (vc31);
\draw[-] (vc13) to (vco13);
\draw[-] (vc31) to (vco31);
\draw[-] (vco13) to (vco13p);
\draw[-] (vco31) to (vco31p);
\draw[-] (vco13) to (vco31);

\draw[-] (vc24) to (vc42);
\draw[-] (vc24) to (vco24);
\draw[-] (vc42) to (vco42);
\draw[-] (vco24) to (vco24p);
\draw[-] (vco42) to (vco42p);
\draw[-] (vco24) to (vco42);

\draw[-] (vc53) to (vc35);
\draw[-] (vc53) to (vco53);
\draw[-] (vc35) to (vco35);
\draw[-] (vco53) to (vco53p);
\draw[-] (vco35) to (vco35p);
\draw[-] (vco53) to (vco35);

\draw[-] (vco14) edge node {$M_1^a$} (a14);
\draw[-] (vco14) edge node [swap]{$M_1^b$} (b14);
\draw[-] (vco41) edge node [swap]{$M_4^a$} (a41);
\draw[-] (vco41) edge node {$M_4^b$} (b41);

\draw[-] (vco13) edge node [swap]{$M_1^a$} (a13);
\draw[-] (vco13) edge node {$M_1^b$} (b13);
\draw[-] (vco31) edge node {$M_3^a$} (a31);
\draw[-] (vco31) edge node [swap]{$M_3^b$} (b31);

\draw[-] (vco24) edge node {$M_2^a$} (a24);
\draw[-] (vco24) edge node [swap]{$M_2^b$} (b24);
\draw[-] (vco42) edge node [swap]{$M_4^a$} (a42);
\draw[-] (vco42) edge node {$M_4^b$} (b42);

\draw[-] (vco53) edge node [swap]{$M_5^a$} (a53);
\draw[-] (vco53) edge node {$M_5^b$} (b53);
\draw[-] (vco35) edge node {$M_3^a$} (a35);
\draw[-] (vco35) edge node [swap]{$M_3^b$} (b35);

\end{tikzpicture}
\caption{\label{GraphGl1}Part one of the graph $G'$ resulting from the reduction applied to graph $G$ of the Figure \ref{GraphG} regarding the 3-Multicolor Clique Problem}
\end{figure}
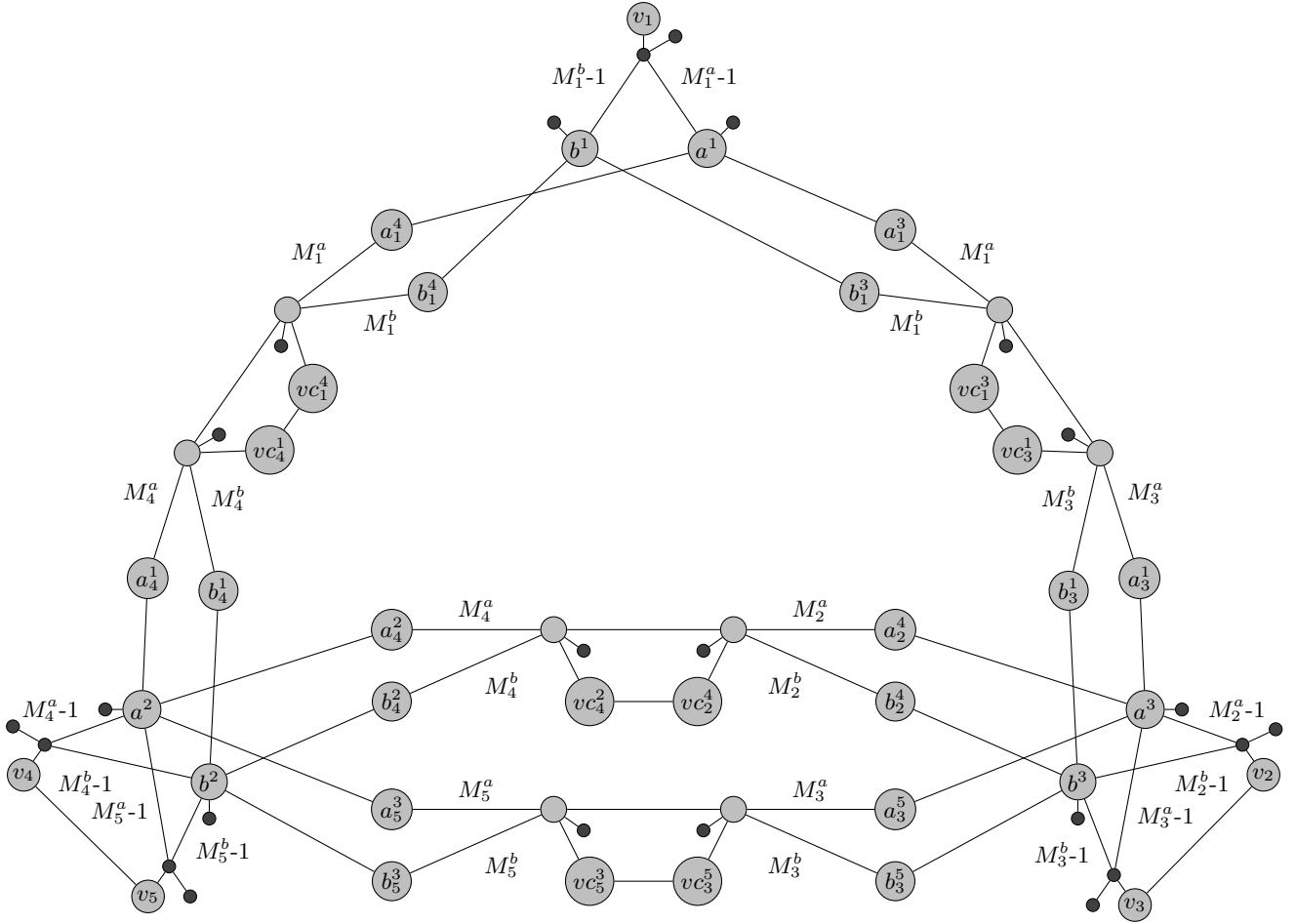

\begin{figure}[ht]
\centering
\begin{tikzpicture}[scale=0.9,auto]

\def \angle {40}
\def \angledes {130}
\def \angled {25}
\def \dist {5}
\def \distd {2}

\tikzstyle{vertex}=[draw,circle,fill=black!25,minimum size=10pt,inner sep=1pt]
\tikzstyle{pendant}=[draw,circle,fill=black!75,minimum size=5pt,inner sep=1pt]

\node[vertex] (l14) at (0*\angle+\angledes:\dist) {$l_1^4$};
\node[vertex,shift={(0*\angle+\angledes+\angled:\distd)}] (a14) at (l14) {$a_1^4$};
\node[vertex,shift={(0*\angle+\angledes-\angled:\distd)}] (b14) at (l14) {$b_1^4$};

\node[vertex] (l41) at (1*\angle+\angledes:\dist) {$l_4^1$};
\node[vertex,shift={(1*\angle+\angledes+\angled:\distd)}] (a41) at (l41) {$a_4^1$};
\node[vertex,shift={(1*\angle+\angledes-\angled:\distd)}] (b41) at (l41) {$b_4^1$};

\node[vertex] (l13) at (2*\angle+\angledes:\dist) {$l_1^3$};
\node[vertex,shift={(2*\angle+\angledes+\angled:\distd)}] (a13) at (l13) {$a_1^3$};
\node[vertex,shift={(2*\angle+\angledes-\angled:\distd)}] (b13) at (l13) {$b_1^3$};

\node[vertex] (l31) at (3*\angle+\angledes:\dist) {$l_3^1$};
\node[vertex,shift={(3*\angle+\angledes+\angled:\distd)}] (a31) at (l31) {$a_3^1$};
\node[vertex,shift={(3*\angle+\angledes-\angled:\distd)}] (b31) at (l31) {$b_3^1$};

\node[vertex] (l24) at (4*\angle+\angledes:\dist) {$l_2^4$};
\node[vertex,shift={(4*\angle+\angledes+\angled:\distd)}] (a24) at (l24) {$a_2^4$};
\node[vertex,shift={(4*\angle+\angledes-\angled:\distd)}] (b24) at (l24) {$b_2^4$};

\node[vertex] (l42) at (5*\angle+\angledes:\dist) {$l_4^2$};
\node[vertex,shift={(5*\angle+\angledes+\angled+5:\distd)}] (a42) at (l42) {$a_4^2$};
\node[vertex,shift={(5*\angle+\angledes-\angled:\distd)}] (b42) at (l42) {$b_4^2$};

\node[vertex] (l53) at (6*\angle+\angledes:\dist) {$l_5^3$};
\node[vertex,shift={(6*\angle+\angledes+\angled:\distd)}] (a53) at (l53) {$a_5^3$};
\node[vertex,shift={(6*\angle+\angledes-\angled:\distd)}] (b53) at (l53) {$b_5^3$};

\node[vertex] (l35) at (7*\angle+\angledes:\dist) {$l_3^5$};
\node[vertex,shift={(7*\angle+\angledes+\angled:\distd)}] (a35) at (l35) {$a_3^5$};
\node[vertex,shift={(7*\angle+\angledes-\angled:\distd)}] (b35) at (l35) {$b_3^5$};

\node[vertex] (p) at (0,0) {$p$};
\node[pendant,shift={(270:0.75)}] (pp) at (p) {};
\node[vertex,shift={(90:3)}] (z) at (p) {$z$};
\node[pendant,shift={(0:0.5)}] (zp) at (z) {};

\draw[-] (l14) edge node {$E$-$M_1^a$-2} (a14);
\draw[-] (l14) edge node [swap] {$E$-$M_1^b$-2} (b14);

\draw[-] (l41) edge node {$E$-$M_4^a$-2} (a41);
\draw[-] (l41) edge node [swap] {$E$-$M_4^b$-2} (b41);

\draw[-] (l13) edge node {$E$-$M_1^a$-2} (a13);
\draw[-] (l13) edge node [swap] {$E$-$M_1^b$-2} (b13);

\draw[-] (l31) edge node {$E$-$M_3^a$-2} (a31);
\draw[-] (l31) edge node [swap] {$E$-$M_3^b$-2} (b31);

\draw[-] (l24) edge node {$E$-$M_2^a$-2} (a24);
\draw[-] (l24) edge node [swap] {$E$-$M_2^b$-2} (b24);

\draw[-] (l42) edge node {$E$-$M_4^a$-2} (a42);
\draw[-] (l42) edge node [swap] {$E$-$M_4^b$-2} (b42);

\draw[-] (l53) edge node {$E$-$M_5^a$-2} (a53);
\draw[-] (l53) edge node [swap] {$E$-$M_5^b$-2} (b53);

\draw[-] (l35) edge node {$E$-$M_3^a$-2} (a35);
\draw[-] (l35) edge node [swap] {$E$-$M_3^b$-2} (b35);

\draw[-] (l14) to (p);
\draw[-] (l41) to (p);
\draw[-] (l13) to (p);
\draw[-] (l31) to (p);
\draw[-] (l24) to (p);
\draw[-] (l42) to (p);
\draw[-] (l53) to (p);
\draw[-] (l35) to (p);

\draw[-] (p) to (pp);
\draw[-] (z) to (zp);

\draw[-] (p) edge node {$D$} (z);

\end{tikzpicture}
\caption{\label{GraphGl2}Part two of the graph $G'$ resulting from the reduction applied to graph $G$ of the Figure \ref{GraphG} regarding the 3-Multicolor Clique Problem}
\end{figure}
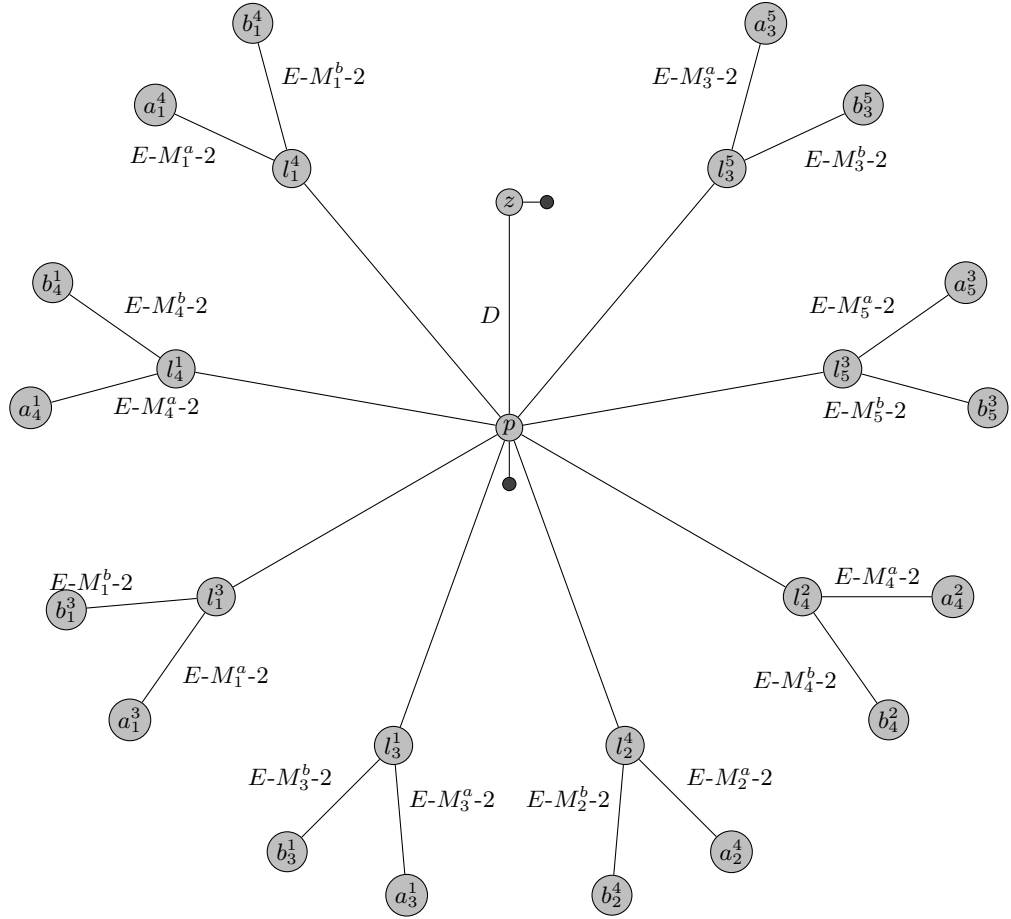

Since the set $S = T \cup \{p\}$ is a separator set of $G'$ and each connected component of $G' - S$ have treewidth at most 2, we have that $tw(G') \leq |S| + (2+1) - 1 = 2k+3$.

Now, let us prove that there is a $k$-clique in $G$ such that all vertices are in different partitions if and only if $t(G') \geq D + E + 1$.

First, assume that there is a $k$-clique in $G$ such that all vertices are in different partitions. Let us build a percolating set $S$ such that $t(S) \geq D+E+1$.

Let $C = \{v_{c_1},v_{c_2},\hdots,v_{c_k}\}$ be the set of vertices in this $k$-clique. Initially, set $S = C$. For each set $VC_i^j$ such that, for some $1 \leq s \leq k$, either $i = c_s$ or $j = c_s$, add to $S$ the vertex $vc_i^j$, if $j = c_s$, or add $vc_j^i$ to $S$, if $i = c_s$. Since $C$ induce a $k$-clique on $G$, we do not have a set $VC_i^j$ such that, for some $1 \leq s,r \leq k$, $i = c_s$ and $j = c_r$. Finally, for each set $VC_i^j$ such that, for all $1 \leq s \leq k$, neither $i = c_s$ nor $j = c_s$, add, arbitrarily, either $vc_i^j$ or $vc_j^i$ to $S$.

We have that $S$ infects each vertex $a^x$ at time $M_{v_{c_x}}^a$ and each vertex $b^x$ at time $M_{v_{c_x}}^b$. Since, for all $1 \leq i \neq j \leq n$, either $M_i^a \geq M_j^a + 2$ and $M_i^b \leq M_j^b - 2$ or $M_i^a \leq M_j^a - 2$ and $M_i^b \geq M_j^b + 2$, then for all vertices $a_i^j$ and $b_i^j$, either $a_i^j$ is infected by $S$ at time at least $M_i^a + 2$ or $b_i^j$ is infected by $S$ at time at least $M_i^b + 2$ and, thus, all vertices $l_i^j \in L$ are infected by $S$ at time at least $E$. Then, we can conclude that $p$ is infected by $S$ at time at least $E+1$ and $z$ is infected by $S$ at time at least $D+E+1$. Therefore, we have that $t(G) \geq D + E + 1$.

Now, assume that there is no $k$-clique in $G$ such that all vertices are in different partitions.

Since is necessary and sufficient to choose one vertex, and only one, from each set $V^x$ and each set $VC_i^j$ so that the resulting set infects the whole graph $G'$, we can focus on proving that the percolating sets built in that way infect $G'$ at time less than $D+E+1$. Let $SS$ be the family of such percolating sets.

By our construction and the values of $D$ and $E$, we have that only $z$ can possible be infected at time at least $D+E+1$ by some percolating set in $SS$ because $z$ is the only vertex in $G'$ such that there is a simple path from some vertex in $S$ to $z$ with length greater than $D+E$, for any $S \in SS$. Let $S$ be an arbitrary percolating set in $SS$. Since there is no $k$-clique in $G$ such that all vertices are in different partitions, there will be vertices $a^x,b^x$ and $vo_i^j$ such that $S$ infects $a^x$ at time $M_i^a$, $b^x$ at time $M_i^b$ and $vo_i^j$ at time 1. Thus, we have that there will be vertices $a_i^j$ and $b_i^j$ that are infected by $S$ at times $M_i^a + 1$ and $M_i^b + 1$, respectively. Then, we have that there will be some vertex $l_i^j$ infected by $S$ at time $E-1$. We can, then, conclude that $p$ is infected by $S$ at time at most $E$ and $z$ at time at most $D+E$. Since $z$ is the only vertex that could be possibly infected at time at least $D+E+1$ by $S$, we have that $t(S) < D+E+1$ and, since $S$ is a arbitrary set of $SS$, therefore, $t(G) < D+E+1$.
\end{proof}

% ----------------------------------------------------------------------------------
% ----------------------------------------------------------------------------------
% ----------------------------------------------------------------------------------
% ----------------------------------------------------------------------------------
The next theorem proves that the Percolation Time Problem is polynomial time solvable for fixed $tw(G)$ and it is linear time solvable for fixed $tw(G)+k$.

\begin{theorem}\label{fpt-tw3}
Let $w$ be an integer and let $G$ be a graph with treewidth $tw(G)=w$.
Then $t(G)$ can be computed in time $O(50^{w+1} n^{w+2})$ and we can decide whether $t(G)\geq k$ in time $O((50k)^{w+1}n)$.
\end{theorem}

In order to prove this theorem, we have to give some definitions first.
A \emph{tree decomposition} \cite{rob-seymour84} of a graph $G$ is a tuple $(\mathcal{T},\mathcal{B})$, where $\mathcal{T}$ is a tree, $\mathcal{B}$ contains a subset $B_t\subseteq V(G)$ for each node $t\in \mathcal{T}$ and:

\begin{enumerate}
\item[(i)] for each vertex $u$ of $G$, there exists some $B_t\in \mathcal{B}$ containing $u$;
\item[(ii)] for each edge $uv$ of $G$, there exists some $B_t\in \mathcal{B}$ containing $u$ and $v$; and 
\item[(iii)] if $B_i$ and $B_j$ both contain a vertex $v$, then, for every node $t$ of the tree in the (unique) path between $i$ and $j$, $B_t$ also contains $v$.
\end{enumerate}

Observe from (iii) that $\mathcal{T}[\{t:v\in B_t\}]$ is connected, for all $v\in V(G)$. That is, the nodes associated with vertex $v$ form a connected component of $\mathcal{T}$. In other words, if $t$ is on the unique path of $\mathcal{T}$ between $i$ to $j$, then $B_i\cap B_j\subseteq B_t$.
The \emph{width of $(\mathcal{T},\mathcal{B})$} is $\max\{|B_t|-1:t\in T, B_t\in\mathcal{B}\}$ and the \emph{treewidth $tw(G)$ of $G$} is the minimum width over all tree decompositions of $G$ \cite{rob-seymour84}.

Consider $\mathcal{T}$ to be rooted in $r$. We say that $(\mathcal{T},\mathcal{B})$ is a \emph{nice tree decomposition} \cite{kloks91} if each node $t$ of $\mathcal{T}$ is either a leaf, or $t$ has exactly two children $t_1$ and $t_2$ with $B_{t}=B_{t_1}=B_{t_2}$ (called \emph{join node}), or $t$ has exactly one child $t'$ and either $B_t=B_{t'}\setminus\{x\}$ (called \emph{forget node}) or $B_t=B_{t'}\cup\{x\}$ (called \emph{introduce node}), for some $x\in V(G)$. It is known that, for a fixed $k$, we can determine if the treewidth of a given graph $G$ is at most $k$, and if so, find a nice tree decomposition of $G$ with $O(n)$ nodes and width at most $w$ in linear time \cite{bodlaender96}.
Given a node $t$ of $\mathcal{T}$, we denote by $G_t$ the subgraph of $G$ induced by $\bigcup_{t'\in V(\mathcal{T}_t)}B_{t'}$, where $\mathcal{T}_t$ is the subtree of $\mathcal{T}$ rooted at $t$.

For each node $t \in \mathcal{T}$, we will compute a table $W_t$ with integer values for every pair $(p,f)$, where $p$ is a function $p: B_t \rightarrow \{0,1,\hdots,n-1\}$ that is an assignment of times for the vertices in $B_t$ and $f$ is a function $f : B_t \rightarrow \{z,o_1,o_2,t_1,t_2\}$ which associate a status to each vertex in $B_t$. Let $v \in B_t$:
\begin{itemize}
\item $f(v) = z$ indicates that $v$ has does not have any neighbor in $G_t$ infected at time less than $p(v)$;
\item $f(v) = o_1$ indicates that $v$ has exactly one neighbor in $G_t$ infected at time less than $p(v)$ and this neighbor is infected at time less than $p(v)-1$;
\item $f(v) = o_2$ indicates that $v$ has exactly one neighbor in $G_t$ infected at time less than $p(v)$ and this neighbor is infected at time equal to $p(v)-1$;
\item $f(v) = t_1$ indicates that $v$ has two or more neighbors in $G_t$ infected at time less than $p(v)$ and exactly one of them is infected at time less than $p(v)-1$;
\item $f(v) = t_2$ indicates that $v$ has two or more neighbors in $G_t$ infected at time less than $p(v)$ and all of them are infected at time equal to $p(v)-1$.
\end{itemize}

Since $|B_t| \leq w+1$, we have that the number of functions $p$ and $f$ are bounded by $n^{w+1}$ and $5^{w+1}$, respectively.

Let $f$ be a function $V(G) \rightarrow \{0,1,\hdots,n-1\}$. We say that $f$ is a \emph{infection time function} of $G$ if there is a percolating set $S$ such that $t(S,v) = f(v)$ for all $v \in V(G)$. We say that $f$ is a \emph{quasi-infection time function} of $G$ if, for each $v \in V(G)$, there is at most one vertex $u \in N(v)$ such that $f(u) < f(v) - 1$.

Let $W_t(p,f) = \max_g \max_{v \in V(G_t)} g(v)$ where $g$ iterates over all quasi-infection time functions of $G_t$ such that: each vertex $v\in V(G_t)\setminus B_t$, where $g(v) > 0$, has at least two neighbors $u_1$ and $u_2$ such that $g(u_1) < g(v)$ and $g(u_2) < g(v)$; $p$ is $g$ restricted to $B_t$ and $g$ does not contradict $f$ regarding the vertices in $B_t$. We say that $g$ contradicts $f$ if, for some $v \in B_t$, $f(v)$ does not correspond to the reality of $g$ on $v$ and its neighbors in $G_t$. Set $W_t(p,f) = -1$ if there is no quasi-infection time function $g$ satisfying these conditions.

Our algorithm computes $W_t(p,f)$ for every $t\in T$ and every pair $(p,f)$ of functions. Clearly, $t(G)$ is equal to the highest $W_r(p,f)$ (recall that $r$ is the root of $\mathcal{T}$) such that, for all vertices $v \in B_r$, either $f(v) = t_1$ or $f(v) = t_2$ or $p(v) = 0$.

For some indexes $h=(p,f)$, we can conclude directly that $W_t(h) = -1$ for some $t \in T$. We will say that such indexes are invalid. We are only interested in valid indexes, defined as below. Let $t \in T$ and $h = (p,f)$ be an index of $W_t$. We say that $h$ is valid if and only if, for each vertex $v \in B_t$, we have that:
\begin{itemize}
\item If $f(v) = z$ then $p(z) \geq p(v)$ for all $z \in N(v) \cap B_t$;
\item If $f(v) = o_1$ then there is at most one vertex $z \in N(v) \cap B_t$ such that $p(z) < p(v)$ and, if there is one, we have that $p(z) < p(v)-1$;
\item If $f(v) = o_2$ then there is at most one vertex $z \in N(v) \cap B_t$ such that $p(z) < p(v)$ and, if there is one, we have that $p(z) = p(v)-1$;
\item If $f(v) = t_1$ then there is at most one vertex $z \in N(v) \cap B_t$ such that $p(z) < p(v) - 1$;
\item If $f(v) = t_2$ then all vertices $z \in N(v) \cap B_t$ where $p(z) < p(v)$, are such that $p(z) = p(v)-1$.
\end{itemize}

Clearly, if $h=(p,f)$ is an invalid index then $W_t(h) = -1$ because any extension of $p$ for $G_t$ would contradict $f$. Now, we will determine the value of $W_t(h)$ depending on the type of the node $t$.

\begin{lemma}
If $t$ is a leaf node and $h = (p,f)$ is an index of $W_t$, then we can compute $W_t(h)$ in $O(w^2)$ time.
\end{lemma}

\begin{proof}
We have that $W_t(h) = \max_{v \in B_t} p(v)$ if and only if $h$ is valid, which can be checked in $O(w^2)$ time.
\end{proof}

\begin{lemma}
Let $t$ be a forget node with child $t'$ such that $B_t = B_{t'} \setminus \{v\}$. Let $h = (p,f)$ be a valid index of $W_t$. Then $W_t(h) = \max_{h'} W_{t'}(h')$ where $h' = (p',f')$ iterates over all $p'$ and $f'$ such that $p'$ and $f'$ are extensions of, respectively, $p$ and $f$ for $B_{t'}$ and either $f'(v) = t_1$ or $f'(v) = t_2$ or $p'(v) = 0$.
\end{lemma}

\begin{proof}

First suppose that $\max_{h'} W_{t'}(h') = -1$, i.e., for each valid index $h'=(p',f')$ as stated in the Lemma, $W_{t'}(h') = -1$. Then, since $G_t = G_{t'}$, $B_t \cup \{v\} = B_{t'}$ and each $p'$ and $f'$ are extensions of, respectively, $p$ and $f$ for $B_{t'}$, we have that there is no quasi-infection time function $g$ of $G_t$ such that each vertex $u\in V(G_t)\setminus B_t$, where $g(u) > 0$, has at least two neighbors $u_1$ and $u_2$ such that $g(u_1) < g(u)$ and $g(u_2) < g(u)$; $p$ is $g$ restricted to $B_t$ and $g$ does not contradict $f$ regarding the vertices in $B_t$. Therefore, $W_t(h) = -1 = \max_{h'} W_{t'}(h')$.

Now, suppose that $\max_{h'} W_{t'}(h') > -1$. Let us prove that $W_t(h) = \max_{h'} W_{t'}(h')$.

First, we are going to prove that $W_t(h) \geq \max_{h'} W_{t'}(h')$. Let $i = (x,y)$ be an index that realizes the previous maximum and let $g$ be the quasi-infection time function of $G_{t'}$ such that $W_{t'}(i) = \max_{r \in V(G_{t'})} g(r)$ and each vertex $r\in V(G_t)\setminus B_t$, where $g(r) > 0$, has at least two neighbors $u_1$ and $u_2$ such that $g(u_1) < g(r)$ and $g(u_2) < g(r)$; $x$ is $g$ restricted to $B_t$ and $g$ does not contradict $y$ regarding the vertices in $B_t$.

Since $g$ does not contradict $y$ and $G_t = G_{t'}$, we have that each vertex $r \in \{v\} \cup (V(G_t') \setminus B_t') = V(G_t) \setminus B_t$ either has at least two neighbors $z_1$ and $z_2$ such that $g(z_1) < g(r)$ and $g(z_2) < g(r)$ or $g(r) = 0$. 

Also, since $x$ is an extension of $p$ for $B_{t'}$ and $x$ is $g$ restricted to $B_{t'}$ and $B_t \subseteq B_{t'}$, we have that $p$ is $g$ restricted to $B_t$.

Additionally, since $g$ does not contradict $y$ regarding the vertices in $B_{t'}$ and $y$ is an extension of $f$ for $B_{t'}$, $g$ does not contradict $f$ regarding the vertices in $B_t$.

Therefore, we have that $W_t(h) = \max_{g'} max_{r \in V(G_t)} g'(r) \geq \max_{r \in V(G_t)} g(r) = W_{t'}(i) = \max_{h'} W_{t'}(h')$.

Now, let us prove that $W_t(h) \leq \max_{h'} W_{t'}(h')$. We have that $W_t(h) = \max_g' \max_{v \in V(G_t)} g'(v)$. Let $g$ be a quasi-infection time function of $G_t$ that realizes this maximum. Let $p'$ be $g$ restricted to $B_{t'}$, which is an extension of $p$ for $B_{t'}$, and $f'$ be set accordingly to $g$, which is an extension of $f$ for $B_{t'}$.

Since $v \in G_t \setminus B_t$ and $f'$ was set accordingly to $g$, we have that either $f'(v)=t_1$ or $f'(v)=t_2$ or $p'(v) = 0$.

Thus, $g$ is also a quasi-time infection function of $G_{t'}$ that realizes $\max_{g'} max_{r \in V(G_{t'})} g'(r)$ and, since also each vertex $r \in V(G_{t'}) \setminus B_{t'}$, where $p'(r) > 0$, has at least two neighbors $z_1$ and $z_2$ such that $g(z_1) < g(r)$ and $g(z_2) < g(r)$; $p'$ is $g$ restricted to $B_{t'}$ and $g$ does not contradict $f'$ regarding the vertices in $B_{t'}$, then $W_{t'}(p',f') = \max_{r \in V(G_t)} g(r)$.

Therefore, we have that $W_t(h) = \max_{r \in V(G_t)} g(r) = W_{t'}(p',f') \leq \max_{h'} W_{t'}(h')$.
\end{proof}

\begin{lemma}
Let $t$ be a introduce node with child $t'$ such that $B_t = B_{t'} \cup \{v\}$. Let $h = (p,f)$ be a valid index of $W_t$. Also, let $p'$ be $p$ restricted to $B_{t'}.$ Then, if $\max_{f'} (W_{t'}(p',f')) > -1$, $W_t(h) = \max (p(v), \max_{f'} (W_{t'}(p',f'))$, where $f'$ iterates over all functions $f' : B_{t'} \rightarrow \{z,o_1,o_2,t_1,t_2\}$ such that, for all $r \in B_{t'} \cap N(v)$ where $p(r) > p(v)$, we have:
\begin{enumerate}
\item If $p(v) < p(r)-1$ and $f(r) = o_1$ then $f'(r) = z$;
\item If $p(v) = p(r)-1$ and $f(r) = o_2$ then $f'(r) = z$;
\item If $p(v) = p(r)-1$ and $f(r) = t_1$ then either $f'(r) = o_1$ or $f'(r) = t_1$;
\item If $p(v) < p(r)-1$ and $f(r) = t_1$ then either $f'(r) = o_2$ or $f'(r) = t_2$;
\item If $p(v) = p(r)-1$ and $f(r) = t_2$ then either $f'(r) = o_2$ or $f'(r) = t_2$;
\end{enumerate}

And for all other vertices $r$, $f'(r) = f(r)$. If, for all such $f'$, $\max_{f'} (W_{t'}(p',f')) = -1$, then $W_t(h) = -1$.
\end{lemma}

\begin{proof}
First suppose that $\max_{f'} W_{t'}(p',f') = -1$, then we have that there is no quasi-infection time function $g$ of $G_{t'}$ that extends $p'$ and respects some $f'$ where $f'$ is as in the Lemma.

Since $V(G_t) = V(G_{t'}) \cup \{v\}$, $B_t = B_{t'} \cup \{v\}$ and $p'$ is $p$ restricted to $B_{t'}$, by a simple case analysis on $f'$, we can conclude that there is no quasi-infection time function $g$ of $G_t$ that extends $p$ and does not contradict $f$. Therefore, $W_t(h) = -1$.

Now, suppose that $\max_{f'} W_{t'}(p',f') > -1$. Let us prove that $W_t(h) = \max (p(v), \max_{f'} (W_{t'}(p',f'))$.

First, we are going to prove that $W_t(h) \geq \max (p(v), \max_{f'} (W_{t'}(p',f'))$. Let $y : B_{t'} \rightarrow \{z,o_1,o_2,t_1,t_2\}$ be a function that realizes the previous maximum ($\max_{f'} (W_{t'}(p',f')$) and let $g'$ be the quasi-infection time function of $G_{t'}$ such that $W_{t'}(p',y) = max_{r \in V(G_{t'})} g'(r)$ where each vertex $r \in V(G_{t'}) \setminus B_{t'}$, where $g'(r) > 0$, has at least two neighbors $z_1$ and $z_2$ such that $g'(z_1) < g'(r)$ and $g'(z_2) < g'(r)$; $p'$ is $g'$ restricted to $B_{t'}$ and $g'$ does not contradict $y$ regarding the vertices in $B_{t'}$.

Let $g$ be an extension of $g'$ for $B_t$ such that $g(v) = p(v)$. So, we have that $g$ is a quasi-infection time function of $G_t$ where each vertex $r \in V(G_t) \setminus B_t$, where $g(r) > 0$, has at least two neighbors $z_1$ and $z_2$ such that $g(z_1) < g(r)$ and $g(z_2) < g(r)$; $p$ is $g$ restricted to $B_t$ and, by a case analysis on $y$, $g$ does not contradict $f$ regarding the vertices in $B_t$.

Therefore, we have that $W_t(h) = \max_{g'} max_{r \in V(G_t)} g'(r) \geq \max_{r \in V(G_t)} g(r) = \max (p(v),W_{t'}(p',y)) = \max (p(v),\max_{f'} W_{t'}(p',f'))$.

Now, let us prove that $W_t(h) \leq \max (p(v),\max_{f'} W_{t'}(p',f'))$. We have that $W_t(h) = \max_g' \max_{r \in V(G_t)} g'(r)$. Let $g$ be a quasi-infection time function of $G_t$ that realizes this maximum. Let $p'$ be $g$ restricted to $B_{t'}$, which implies that $p'$ is $p$ restricted to $B_{t'}$, and $y : B_{t'} \rightarrow \{z,o_1,o_2,t_1,t_2\}$ be set accordingly to $g$ restricted to $B_{t'}$, which, by a simple case analysis on $f$ and $p$, will fall in one of the cases in the Lemma.

Thus, at least one of the two cases occur:
\begin{enumerate}
\item $\max_{r \in V(G_t)} g(r) = g(v)$
\item $\max_{r \in V(G_t)} g(r) = max_{r \in V(G_{t'})} g(r)$.
\end{enumerate}

If $g(v) = \max_{r \in V(G_t)} g(r)$ then we have that $W_t(h) = \max_{r \in V(G_t)} g(r) = g(v) = p(v) \leq \max (p(v),\max_{f'} W_{t'}(p',f'))$.

On the other hand, if $\max_{r \in V(G_t)} g(r) = max_{r \in V(G_{t'})} g(r)$, then $g$ is also a quasi-time infection function of $G_{t'}$ that realizes $\max_{g'} max_{r \in V(G_{t'})} g'(r)$ and, since also each vertex $r \in V(G_{t'}) \setminus B_{t'}$, where $g(r) > 0$, has at least two neighbors $z_1$ and $z_2$ such that $g(z_1) < g(r)$ and $g(z_2) < g(r)$; $p'$ is $g$ restricted to $B_{t'}$ and $g$ does not contradict $y$ regarding the vertices in $B_{t'}$, then $W_{t'}(p',y) = \max_{r \in V(G_{t'})} g(r)$ and, therefore, we have that $W_t(h) = \max_{r \in V(G_t)} g(r) = W_{t'}(p',y) \leq \max_{f'} W_{t'}(p',f') \leq \max (p(v),\max_{f'} W_{t'}(p',f'))$.
\end{proof}

\begin{lemma}
Let $t$ be a join node with children $t_1$ and $t_2$ and let $h = (p,f)$ be a valid index of $W_t$. If there is some pair $(f_1,f_2)$ such that $\min(W_{t_1}(p,f_1),W_{t_2}(p,f_2)) > -1$, then $W_t(h) = \max_{(f_1,f_2)} \max(W_{t_1}(p,f_1),W_{t_2}(p,f_2))$ where the pair $(f_1,f_2)$ iterates over all pair of functions $f_1 : B_{t_1} \rightarrow \{z,o_1,o_2,t_1,t_2\}$ and $f_2 : B_{t_2} \rightarrow \{z,o_1,o_2,t_1,t_2\}$ such that, for each $r \in B_t$, there are $i \neq j \in \{1,2\}$ where:
\begin{enumerate}
\item If $f(r) = z$ then $f_i(r) = z$ and $f_j(r) = z$;
\item If $f(r) = o_1$ then $f_i(r) = z$ and $f_j(r) = o_1$;
\item If $f(r) = o_2$ then $f_i(r) = z$ and $f_j(r) = o_2$;
\item If $f(r) = t_1$ then either:
\begin{enumerate}
\item $f_i(r) = z$ and $f_j(r) = t_1$; or
\item $f_i(r) = o_1$ and $f_j(r) = o_2$; or
\item $f_i(r) = o_1$ and $f_j(r) = t_2$; or
\item $f_i(r) = o_2$ and $f_j(r) = t_1$; or
\item $f_i(r) = t_1$ and $f_j(r) = t_2$.
\end{enumerate}
\item If $f(r) = t_2$ then either:
\begin{enumerate}
\item $f_i(r) = z$ and $f_j(r) = t_2$; or
\item $f_i(r) = o_2$ and $f_j(r) = o_2$; or
\item $f_i(r) = o_2$ and $f_j(r) = t_2$; or
\item $f_i(r) = t_2$ and $f_j(r) = t_2$.
\end{enumerate}
\end{enumerate}
If there is no pair $(f_1,f_2)$ such that $\min(W_{t_1}(p,f_1),W_{t_2}(p,f_2)) > -1$, then $W_t(h) = -1$
\end{lemma}

\begin{proof}
First, suppose that there is no pair $(f_1,f_2)$ such that $\min(W_{t_1}(p,f_1),W_{t_2}(p,f_2)) > -1$ where $f_1$ and $f_2$ are as stated in the Lemma. Thus, we have that, for all pairs $(f_1,f_2)$ as stated in the Lemma, there is a $i$ in $\{1,2\}$, such that there is no quasi-infection time function $g_i$ of $G_{t_i}$ such that each vertex $v \in V(G_{t_i}) \setminus B_{t_i}$, where $g_i(v) > 0$, has at least two neighbors $z_1$ and $z_2$ such that $g_i(z_1) < g_i(v)$ and $g_i(z_2) < g_i(v)$; $p$ is $g_i$ restricted to $B_{t_i}$ and $g_i$ does not contradict $f_i$ regarding the vertices in $B_{t_i}$.

Suppose, by contradiction, that there is a quasi-infection time function $g$ of $G_t$ such that each vertex $v \in V(G_t) \setminus B_t$, where $g(v) > 0$, has at least two neighbors $z_1$ and $z_2$ such that $g(z_1) < g(v)$ and $g(z_2) < g(v)$; $p$ is $g$ restricted to $B_t$ and, for each vertex $v \in B_t$, $g$ does not contradict $f$ regarding the vertices in $B_t$. Letting $g_1$ be $g$ restricted to $G_{t_1}$ and $g_2$ be $g$ restricted to $G_{t_2}$ and $(f_1,f_2)$ be as stated in the Lemma, basing on $f$, we have that for all $i \in \{1,2\}$, $g_i$ is a quasi-infection time function of $G_{t_i}$ such that each vertex $v \in V(G_{t_i}) \setminus B_{t_i}$, where $g_i(v) > 0$, has at least two neighbors $z_1$ and $z_2$ such that $g_i(z_1) < g_i(v)$ and $g_i(z_2) < g_i(v)$; $p$ is $g_i$ restricted to $B_{t_i}$ and $g_i$ does not contradict $f_i$ regarding the vertices in $B_{t_i}$. However, this contradicts the former paragraph, and, hence, we have that there is no quasi-infection time function $g$ of $G_t$ such that each vertex $v \in V(G_t) \setminus B_t$, where $g(v) > 0$, has at least two neighbors $z_1$ and $z_2$ such that $g(z_1) < g(v)$ and $g(z_2) < g(v)$; $p$ is $g$ restricted to $B_t$ and, for each vertex $v \in B_t$, $g$ does not contradict $f$ regarding the vertices in $B_t$. Therefore, $W_t(h) = -1$.

Now, suppose that there is some pair $(f_1,f_2)$ such that $\min(W_{t_1}(p,f_1),W_{t_2}(p,f_2)) > -1$.

First, let us prove that $W_t(h) \geq \max_{(f_1,f_2)} \max(W_{t_1}(p,f_1),W_{t_2}(p,f_2))$. Let $(y_1,y_2)$ be a pair of functions that realizes the previous maximum ($\max_{(f_1,f_2)} \max(W_{t_1}(p,f_1),W_{t_2}(p,f_2))$). Also, for all $i \in \{1,2\}$, let $g_i$ be the quasi-infection time function of $G_{t_i}$ such that $W_{t_i}(p,y_i) = max_{r \in V(G_{t_i})} g_i(r)$ where each vertex $r \in V(G_{t_i}) \setminus B_{t_i}$, where $g_i(r) > 0$, has at least two neighbors $z_1$ and $z_2$ such that $g_i(z_1) < g_i(r)$ and $g_i(z_2) < g_i(r)$; $p$ is $g_i$ restricted to $B_{t_i}$ and $g_i$ does not contradict $y_i$ regarding the vertices in $B_{t_i}$.

Let $g$ be a quasi-infection time function of $G_t$ where $g(v) = g_1(v)$, if $v \in V(G_{t_1})$, and $g(v) = g_2(v)$, if $v \in V(G_{t_2})$. Note that, since, for all $v \in V(G_{t_1}) \cap V(G_{t_2}) = B_t$, $g_1(v) = g_2(v) = p(v)$, then $g$ is well defined.

We have that $g$ is a quasi-infection time function of $G_t$ such that each vertex $r \in V(G_t) \setminus B_t$, where $g(r) > 0$, has at least two neighbors $z_1$ and $z_2$ such that $g(z_1) < g(r)$ and $g(z_2) < g(r)$; $p$ is $g$ restricted to $B_t$ and, by a simple case analysis on $y_1$ and $y_2$, we have that $g$ does not contradict $f$ regarding the vertices in $B_t$.

Therefore, $W_t(h) = \max_{g'} \max_{r \in V(G_t)} g'(r) \geq \max_{r \in V(G_t)} g(r) = \max(W_{t_1}(p,y_1),W_{t_2}(p,y_2)) = \\ \max_{(f_1,f_2)} \max(W_{t_1}(p,f_1),W_{t_2}(p,f_2))$.

Now, let us prove that $W_t(h) \leq \max_{(f_1,f_2)} \max(W_{t_1}(p,f_1),W_{t_2}(p,f_2))$. We have that $W_t(h) = \max_g' \max_{r \in V(G_t)} g'(r)$. Let $g$ be a quasi-infection time function of $G_t$ that realizes this maximum. Let $g_1$ be $g$ restricted to $G_{t_1}$ e $g_2$ be $g$ restricted to $G_{t_2}$. Also, let $y_1 : B_{t_1} \rightarrow \{z,o_1,o_2,t_1,t_2\}$ be set accordingly to $g_1$ restricted to the graph $G_{t_1}$ and $y_2 : B_{t_2} \rightarrow \{z,o_1,o_2,t_1,t_2\}$ be set accordingly to $g_2$ restricted to the graph $G_{t_2}$. By case analysis on $f$, it is easy to see that $y_1$ and $y_2$ will fall in one of the cases in the Lemma.

Thus, we have that, for all $i \in \{1,2\}$, $g_i$ is a quasi-infection time functions of $G_{t_i}$ such that each vertex $v \in V(G_{t_i}) \setminus B_{t_i}$, where $g_i(v) > 0$, has at least two neighbors $z_1$ and $z_2$ such that $g_i(z_1) < g_i(v)$ and $g_i(z_2) < g_i(v)$; $p$ is $g_i$ restricted to $B_{t_i}$ and $g_i$ does not contradict $y_i$ regarding the vertices in $B_{t_i}$.

Additionally, note that $\max_{r \in V(G_t)} g(r) = \max_{r \in V(G_{t_1})} g_1(r)$ or $\max_{r \in V(G_t)} g(r) = \max_{r \in V(G_{t_2})} g_2(r)$ or possibly both.

Therefore, we have that $W_t(h)=\max_{r \in V(G_t)} g(r) = \max(\max_{r \in V(G_{t_1})} g_1(r),\max_{r \in V(G_{t_2})} g_2(r)) \leq \\ \max(W_{t_1}(p,y_1),W_{t_2}(p,y_2)) \leq \max_{(f_1,f_2)} \max(W_{t_1}(p,f_1),W_{t_2}(p,f_2))$.
\end{proof}

\begin{proof}[Proof of Theorem \ref{fpt-tw3}]
For each node $t$, we have that $W_t$ has $5^{w+1} \cdot n^{w+1}$ indexes. Each index can be computed by checking, in a join node, $O(10^{w+1})$ indexes of his children, in a introduce node, $O(1)$ indexes of his child, and, in a forget node, $O(n)$ indexes of his child. However, one can reduce the time needed to compute the entire table of the forget node simple by computing its table while computing the table of its child, i.e., if $t$ is the forget node and $t'$ is its child, we can initialize $W_t(h)$ with $-1$ for every index $h$ of $W_t$ and every time we compute the value $W_{t'}(p,f)$, letting $p'$ and $f'$ be, respectively, $p$ and $f$ restricted to $B_t$, if  $W_{t'}(p,f) > W_t(p',f')$ and either $f(v) = t_1$ or $f(v) = t_2$ or $p(v) = 0$, then we update the value of $W_t(p',f')$ with the value $W_{t'}(p,f)$. This would add only $O(1)$ time to compute each index of its child's table. Thus, we can find $t(G)$ in $O(n \cdot 10^{w+1} \cdot (5^{w+1}\cdot n^{w+1})) = O(50^{w+1} \cdot n^{w+2})$ time.

In a very similar way, considering that we only have to check the indexes where $p$ is a function $p : B_t \rightarrow \{0,1,\hdots,k\}$, we can show that, for each node $t$, we have that $W_t$ has $5^{w+1} \cdot k^{w+1}$ indexes and each index can be checked worst case in $O(10^{w+1})$, which happens when $t$ is a join node, and then we can find $t(G)$ in $O(n \cdot 10^{w+1} \cdot (5^{w+1} \cdot k^{w+1})) = O((50k)^{w+1}n)$ time.
\end{proof}

% ----------------------------------------------------------------------------------
% ----------------------------------------------------------------------------------
% ----------------------------------------------------------------------------------
% ----------------------------------------------------------------------------------
\section{Acknowledgments}

The statements of some of the results of this paper appeared in WG-2015 (Workshop on Graph-Theoretic Concepts in Computer Science). This research was partially supported by CNPq (Universal Proc. 478744/2013-7) and FAPESP (Proc. 2013/03447-6).

% ----------------------------------------------------------------------------------
% ----------------------------------------------------------------------------------
% ----------------------------------------------------------------------------------
% ----------------------------------------------------------------------------------

\bibliography{wg15thiagomarcilon.bbl}
\end{document}